\documentclass[a4paper,12pt]{article}
\usepackage[T1]{fontenc}
\pdfoutput=1

\usepackage{url}
\usepackage{jheppub}
\usepackage{macro}
\usepackage{appendix}
\usepackage{comment}
\usepackage{array}
\usepackage{yfonts}
\usepackage{comment}
\usepackage{amsthm}
\usepackage{tikz-cd}
\usepackage{amsmath,amssymb}
\usepackage{tikz}
\usepackage{pgfkeys}
\usepackage{pgfopts}
\usepackage{hyperref}
\usepackage{xcolor}
\usepackage{young}
\usepackage{youngtab}
\usepackage{ytableau}
\usepackage{titlesec}
\usetikzlibrary{decorations.pathreplacing}

\usepackage{tikz}

\titleformat{\paragraph}
{\normalfont\normalsize\bfseries}{\theparagraph}{1em}{}
\titlespacing*{\paragraph}
{0pt}{3.25ex plus 1ex minus .2ex}{1.5ex plus .2ex}

\newtheorem{Theorem}{Theorem}[section]

\newtheorem{Conjecture}{Conjecture}[section]
\newtheorem{Definition}{Definition}[section]

\usetikzlibrary{matrix,calc,positioning,decorations.markings,decorations.pathmorphing,decorations.pathreplacing
}
\usetikzlibrary{arrows,cd,shapes}

\title{\textbf Duality, asymptotic charges and higher form symmetries in $p$-form gauge theories}

\author[a,b]{Federico Manzoni}

\affiliation[a]{Mathematics and Physics department, Roma Tre, Via della Vasca Navale 84, Rome, Italy}
\affiliation[b]{INFN Roma Tre Section, Physics department, Via della Vasca Navale 84, Rome, Italy}
\emailAdd{federico.manzoni@uniroma3.it, federico13.manzoni97@gmail.com, ORCID ID: 0000-0002-9979-6154}

\abstract{The surface charges associated with $p$-form gauge fields in the Bondi patch of $D$-dimensional Minkowski spacetime are computed. We show that, under the Hodge duality between the field strengths of the dual formulations, electric-like charges for $p$-forms are mapped to magnetic-like charges for the dual $q$-forms, with $q=D-p-2$. We observe that the complex combination of electric-like and magnetic-like charges transforms under duality according to a specific M\"obius transformation. This leads to a possible construction of CCFT in $D=4$ as a M\"obius-principal equivariant bundle, together with its associated bundles, in order to recover celestial operators.\\
We prove an existence and uniqueness theorem for the duality map relating the asymptotic electric-like charges of the dual descriptions, and we provide an algebraic-topological interpretation of this map. As a result, the duality map has a topological nature and ensures that the charge of one formulation contains information about the dual formulation, leading to a deeper understanding of gauge theories, the non-trivial charges associated with them, and the duality of their observables. Moreover, we propose a link between higher-form symmetry charges, naturally associated with a $p$-form gauge theory, and their asymptotic charges. The higher-form charges are reproduced by choosing the gauge parameter to be constant and supported only on an appropriate codimension submanifold. This could partially answer an open question in the celestial holography program.}

\begin{document}

\maketitle

\section{Introduction}

Gauge theories constitute the foundational language through which modern physics
describes interactions and symmetries across a remarkably wide range of energy scales.
From the Maxwellian formulation of classical electrodynamics \cite{max} to the non-abelian
Yang-Mills theories of the Standard Model, gauge principles organise physical
phenomena by encoding redundancies, local symmetries, and conserved charges in a
single unifying geometric framework \cite{YangMills1954,Jackson1998,Peskin:1995ev}. 
Yet, over the past decades it has become increasingly clear that many essential
aspects of gauge theories only fully emerge in the infrared and at the boundaries of
spacetime \cite{PhysRev.128.2851,Sachs:1962wk,Bondi:1962px,Weinberg:1964ew,Weinberg:1965nx,Strominger:2013jfa,strominger2018lectures}. The systematic study of \emph{asymptotic symmetries}, i.e. symmetries acting non-trivially on a codimension-2 surface, has revealed deep connections between infrared structures, memory effects, conserved charges, and soft theorems
\cite{Strominger_2016,Strominger:2014pwa,Strominger:2017aeh,Pasterski_2016,Lipstein:2015rxa,Hawking:2016msc,Casali:2014xpa,Hamada_2017,He2015,Sen:2017xjn,Liu:2014vva,Pasterski:2015zua,Jokela:2019apz,Hamada:2017gdg,Afshar_2019,Campoleoni:2019ptc,Tolish:2016ggo,BarnichBrandt2002,Barnich:2010ojg,Barnich_2010,Romoli:2024hlc,Compere:2018aar,Manzoni:2025wyr}. These insights have
profound implications ranging from scattering amplitudes to quantum gravity \cite{StiebergerTaylor2018,Duval:1990hj,Donnay:2023mrd,Raclariu:2021zjz,puhm,Pasterski_2021,Ciambelli:2022vot,Manzoni:2024agc}.

A particularly fertile arena for exploring such structures is provided by
$p$-form gauge theories. \textit{Exotic gauge theories}, were originally introduced in the context of higher-dimensional
field theories and extended objects \cite{Henneaux:1986ht} standard Maxwell theory. $p$-form gauge fields now
appear everywhere: in string theory, where forms couple naturally to
Dp-branes, are part of the string spectrum and appear a whole tower of $p$-form states (and mixed symmetry tensor states) in the tensionless limit \cite{Polchinski1995,green_schwarz_witten_2012,Green:2012pqa,Polchinski:1998rq,Polchinski:1998rr,Bonelli_2003}; in supergravity and supersymmetric theories \cite{CremmerJuliaScherk1978,cecotti_2015}
and in holography, where $p$-form bulk fields map to currents in the boundary dual \cite{Maldacena_1999,lYi:1998icc}. More recently,
$p$-form fields have attracted attention in celestial holography 
\cite{Pasterski2017,StiebergerTaylor2018,Raclariu:2021zjz,puhm,Donnay:2023mrd}, where conformal primary wavefunctions
for differential forms were constructed \cite{Donnay:2022ijr}. Beyond the connections with other fields of physics, gauge fields carrying mixed Young-tableau symmetries were shown to admit consistent free gauge theories, provided suitable
gauge-for-gauge redundancies are introduced \cite{Curtright:1980yk,Labastida:1986gy}. 

Alongside the ubiquity of $p$-forms and exotic gauge theory, the notion of \emph{duality} plays a central role in
modern theoretical physics. In Maxwell gauge theory we encounter the familiar
electric-magnetic duality \cite{delrio2024}, but more generic gauge theories exhibit a similar duality between a $p$-form gauge field and its dual $(D-p-2)$-form gauge field counterpart. Moreover, this kind of duality extends also to mixed symmetry tensors gauge theories \cite{DeserTeitelboim1976,Henneaux:2004jw,Hull:2001iu,Medeiros_2003}. 
Such dualities permeate string theory, supersymmetric field theories, M-theory, gauge theories and the
broader web of gauge/gravity correspondences. A notable example is the dual formulation
of the graviton in terms of a mixed-symmetry tensor, the so-called Curtright field, which captures the same on-shell degrees of freedom as the metric description \cite{Curtright:1980yk,Bekaert_2007}.
On the one hand, these dualities exchange dynamical fields; on the other, they can
intertwine the corresponding conserved charges and non-perturbative objects. This
has motivated a systematic analysis of how asymptotic charges behave under duality,
an issue only partially understood beyond the classical four dimensional electromagnetic case. 

A further layer of structure emerges with the modern notion of 
\emph{generalized global symmetries} \cite{Gaiotto:2014kfa, Bhardwaj:2023wzd,Bhardwaj:2023ayw}. In a $p$-form gauge theory, the integrals of the
field strength over appropriate cycles define conserved higher-form charges, which in
turn correspond to topological defects and charged extended objects. Understanding
how such generalized symmetries manifest at null infinity and how they relate to the
asymptotic symmetry algebra, constitutes a pressing question especially within
celestial holography, where establishing a precise dictionary between bulk charges
and CCFT operators remains an open challenge, so much so that it leads to the third question of the Simons collaboration's open questions on celestial holography \cite{simons}.

The present work studies these interconnected themes. Our
starting point, Section \ref{par2}, presents a systematic analysis of the asymptotic charges associated with 
$p$-form gauge fields in $D$-dimensional Minkowski spacetime in Bondi coordinates. This is largely review material \cite{manz2} although some new insights are added, both physical and mathematical, as well as first proposal for the construction of CCFTs starting from principal bundles. We assume both the physical radiation and Coulomb fall-offs since the first are associated to radiative degrees of freedom while the second to static distribution of charged objects. The resulting \emph{electric-like} asymptotic charges, expressed in terms of the $ru$ components of the field strength,
generalise the familiar charges of $D=4$ electromagnetism while being exhibiting a universal
structure across dimensions and gauge form field degree. Similar computations are presented in \cite{Afshar:2018apx,Esmaeili:2020eua}, however authors focus only on critical dimension $D_c=2p+2$. We then incorporate the dual description in terms of the dual $(D-p-2)$-form
field. By mapping the electric sector of one formulation to the magnetic sector of the
dual, we construct \emph{magnetic-like} charges and show that the pair assembles into
a complexified electromagnetic-like charge generating a complexified $\mathfrak{u}_{\mathbb{C}}(1)$ abelian algebra as in the $D=4$ case. Remarkably, this complex charge 
transforms under duality via a Möbius transformation parametrised by a matrix in
$\mathrm{PGL}(2,\mathbb{C})$. This observation resonates with the algebraic structure appearing
in celestial CFT, where Lorentz transformations act as Möbius maps on the celestial
sphere \cite{Donnay:2023mrd,Pasterski_2021,Raclariu:2021zjz,strominger2018lectures}. The last part of Section \ref{par2} contains novel contributions it explores some steps toward a \textit{geometrization of celestial CFT} inspired by the Möbius structure uncovered in the duality of asymptotic charges. We propose that a
CCFT may be viewed as a $\mathrm{PGL}(2,\mathbb{C})$-principal Möbius-equivariant bundle over the
celestial sphere together with its spin bundle lift, with celestial operators arising as sections of appropriate associated bundles.
This geometric viewpoint naturally accommodates tensorial and spinorial structures,
clarifies how conformal weights arise, and provides a framework for encoding duality
as an equivariant action. Moreover, descendants and operator product expansions may
be organised using jet bundles, suggesting a coherent geometric structure underlying
the full CCFT spectrum. While preliminary, this proposal, which is similar to other proposals for the geometrization of QFT such as factorization algebras and bundles that emerge in twistor theory \cite{CostelloGwilliam2017,CostelloGwilliam2021,adamo2018lectures}, opens promising avenues for
understanding how CCFTs are explicitly realized.

A central result of the paper, formulated in Section \ref{secTHM}, is an \emph{existence and uniqueness theorem} for the
duality map between finite and non-vanishing asymptotic electric-like charges of the two formulations. 
The natural mathematical setting for this result is the de~Rham complex and its
exactness properties on Minkowski spacetime. The triviality of the de Rham
cohomology ensures that the duality map is well defined, invertible, and \emph{topological} in 
nature. This perspective hints at possible breakdowns
of the duality in quantum gravity scenarios where topology change becomes allowed
\cite{Horowitz1989,Manzoni:2022htx,Witten1995CG,greene1997stringtheorycalabiyaumanifolds}. We also discuss the duality when asymptotic charges are divergent/vanishing.

In Section \ref{sec4}, we propose a connection between asymptotic charges of $p$-form gauge theories and the
higher-form symmetry charges naturally associated with $p$-form gauge fields. 
Integrating the field strength over suitable codimension submanifolds produces the
usual generalized charges, while smearing the local currents with appropriate residual
gauge parameters of the dual description yields charges proportional to the asymptotic
ones. This provides evidence that asymptotic symmetries may furnish a local refinement of higher-form symmetries at null infinity,
contributing to the ongoing effort to relate generalized symmetries, infrared
structures in gauge theories and holography \cite{Afshar:2018apx,Lake:2018dqm,Francia:2018jtb,Hofman_2018}.

\section{\texorpdfstring{$p$}{p}-form gauge theories, their duality and their asymptotic charges}\label{par2}
Let us discuss $p$-form gauge theories, their duality and the computation of their asymptotic charges in arbitrary $D$. This material retraces the work \cite{manz2}, but adds details and some new features like Paragraph \ref{CCFT}.
\subsection{\texorpdfstring{$p$}{p}-form gauge theories gauge invariant dynamics}\label{pfgdi}
The gauge-invariant dynamics of a $p$-form is described by the Lagrangian \cite{Henneaux:1986ht}
\begin{equation}
    \mathcal{L}=-\frac{1}{2(p+1)!}H_{\mu_1...\mu_{p+1}}H^{\mu_1...\mu_{p+1}}
    \label{lagrangian}
\end{equation}
where $H$ is the $(p+1)$-form field strength of the $p$-form gauge field $B$. The equations of motion, in coordinates, are
\begin{equation}
\partial_{\mu_{1}}H^{\mu_{1}...\mu_{p+1}}=0.
\end{equation}
The description of $p$-form fields in a gauge-invariant lagrangian framework was motivated by the possibility of a $p$-form electromagnetism, i.e. a gauge theory whose gauge field is a $p$-form. It is interesting to note \cite{Henneaux:1986ht} that the authors show that this extension of standard Maxwell theory is compatible with spacetime locality only if the gauge group is the abelian $\mathrm{U(1)}$. \\
The gauge transformation of the $p$-form is parametrized by a $(p-1)$-form $\epsilon$. Indeed, removing one box from the Young tableau of the $p$-form yields the Young tableau of a $(p-1)$-form. The gauge transformation, in coordinates, is then
\begin{equation}
    \delta_{\epsilon}(B_{\mu_{2}...\mu_{p+1}})=\partial_{[\mu_{2}}\epsilon_{\mu_{3}...\mu_{p+1}]};
    \label{gaugevar}
\end{equation}
This is the first level of the gauge-for-gauge redundancy typical of exotic gauge theories. Indeed, we can eliminate one box from the Young tableau of the $(p-1)$-form gauge parameter, obtaining a $(p-2)$-form that parametrizes the gauge transformation of $\epsilon$. These steps repeat until reaching a gauge-for-gauge level parametrized by a $0$-form, i.e. a scalar. 

In order to discuss asymptotic symmetries at future null infinity of Minkowski spacetime, we introduce Bondi coordinates $(u,r,{x^i})$, where $u := t - r$ and the set ${x^i}$ contains the $D-2$ angular variables parametrizing the $(D-2)$-dimensional sphere $S^{D-2}$ at null infinity. In the following, we will refer to angular indices simply as $i, j, k, \ldots \ $. The Minkowski line element in these coordinates is given by
\begin{equation}
    ds^2=-du^2-2dudr+r^2\gamma_{ij}dx^idx^j \ \ \ i,j=1,...,D-2,
\end{equation}
where $\gamma_{ij}$ is the unit $(D-2)$-sphere metric.
Metric components are
\begin{equation}
   g_{\mu \nu} =\begin{bmatrix}
-1 &-1 & 0\\
-1 & 0 & 0\\
0 & 0 & r^2\gamma_{ij}
\end{bmatrix}, \ \ \ 
 g^{\mu \nu} =\begin{bmatrix}
0 & -1 & 0\\
-1 & 1 & 0\\
0 & 0 & \frac{1}{r^2}\gamma_{ij}^{-1}
\end{bmatrix};
\end{equation}
and the non-vanishing Christoffel symbols are
\begin{equation}
    \Gamma^i_{jr}=\Gamma^i_{rj}=\frac{1}{r}\delta^i_j, \ \ \ \Gamma^u_{ij}=-\Gamma^r_{ij}=r\gamma_{ij}, \ \ \ \Gamma^k_{ij}=\frac{1}{2}\gamma^{kl}(-\partial_l\gamma_{ij}+\partial_j\gamma_{li}+\partial_i\gamma_{jl});
    \label{cribon}
\end{equation}
where we have assumed the Levi-Civita connection, which always exists (and is unique) thanks to the Levi-Civita theorem as long as the metric is non-degenerate.
The transformation from retarded Bondi to polar coordinates is given by
\begin{equation}
    \{x^{\alpha}\}=(u,r,x^i) \mapsto \{x^{a'}\}=(u-r,r,rf_1(\{x^i\}),rf_2(\{x^i\}),...,rf_{D-2}(\{x^i\}))
\end{equation}
where ${x^i}$ is a set of $D-2$ angular variables. For future convenience, we define $\mathbb{I}_{\{x^i\}} := \{\mu \mid x^{\mu} \in \{x^i\}\}$, which is the set of indices such that the corresponding coordinate is an angular variable. In Bondi coordinates, fields acquire additional powers of $r$ due to the Jacobian of the coordinate transformation. In general, for each angular index $\mu \in \mathbb{I}_{\{x^i\}}$ we gain one extra factor of $r$ with respect to the Cartesian component. For the radiation fall-offs, when $s \leq p$ indices are angular ones, we have
\begin{equation}
r^s\mathcal{O}\bigg(r^{-\frac{(D-2)}{2}}\bigg)=\mathcal{O}\bigg(r^{\frac{-(D-2-2s)}{2}}\bigg);
\label{asirad}
\end{equation}
while for the Coulomb fall-offs, when $s\leq p$ indices are angular one, we get
\begin{equation}
r^s\mathcal{O}\bigg(r^{-(D-p-2)}\bigg)=\mathcal{O}\bigg(r^{-(D-p-2-s)}\bigg).
\label{asiCou}
\end{equation}

\subsection{Duality in $p$-form gauge theories}
A Young tableau $\lambda = \{\lambda_1, \ldots , \lambda_{\ell}\}$ with $\ell \leq D$ defines an irreducible representation of $\mathrm{O}(D-2)$ if and only if $n^{(\lambda)}_1 + n^{(\lambda)}_2 \leq D-2$. Here $n^{(\lambda)}_1$ and $n^{(\lambda)}_2$ are the numbers of boxes in the first two columns of the Young tableau. The same condition holds for $\mathrm{SO}(D-2)$. In the case of pseudo-orthogonal and pseudo-special orthogonal groups, the same result applies.
However, different Young tableaux can be related by a duality: given $\lambda$ and $\xi$ such that $n_{1}^{(\lambda)} = D - 2 - n_1^{(\xi)}$ and $n_i^{(\lambda)} = n_i^{(\xi)}$ for all $i > 1$, we say that $\lambda$ and $\xi$ are dual. Moreover, for $\mathrm{SO}(D-2)$ these two Young tableaux define the same irreducible representation \cite{hamermesh1989group}. It is important to underline that the Young tableaux $\lambda$ and $\xi$ are completely uncorrelated as Young tableaux labeling different irreducible representations of $\mathrm{GL}(D)$. From a physical point of view, this means that the particles carrying the irreducible representations $\lambda$ and $\xi$ are not dual off-shell, but they propagate the same degrees of freedom once we go on-shell.\\
Without loss of generality, we can restrict ourselves to the cases $p \leq \big[\frac{D-2}{2}\big]$, since the remaining cases are already accounted for by the dual description. The $p$-form representation of $\mathrm{SO}(D-2)$ has the following Young tableau $\lambda$
\begin{equation}
    \centering
    \begin{tabular}{r@{}l}
    \raisebox{-2.3ex}{$p\left\{\vphantom{\begin{array}{c}~\\[4ex] ~
    \end{array}}\right.$} &
    \begin{ytableau}
    ~                \\
    \none[\vdots]    \\
    ~                \\      
    \end{ytableau}\\[-1.5ex]
    \end{tabular}
\end{equation}
If we now apply the rules given before we can construct a dual Young tableau $\xi$ with $p=D-2-n_1^{(\xi)}$ and $0=n_i^{(\xi)}$, that is
\begin{equation}
    \centering
    \begin{tabular}{r@{}l}
    \raisebox{-3.9ex}{$q\left\{\vphantom{\begin{array}{c}~\\[6ex] ~
    \end{array}}\right.$} &
    \begin{ytableau}
    ~                \\
    ~                \\
     \none[\vdots]    \\
    ~                \\       
    \end{ytableau}\\[-1.5ex]
    \end{tabular}
\end{equation}\\
hence, for a $q$-form to be dual to a $p$-form we need $q=D-2-p \geq \frac{D-2}{2}$. In the following we will call $B$ the $p$-form gauge field and $H= dB$ its field strength, while we will call $\tilde{B}$ the dual $q$-form gauge field and $\tilde{H}=d\tilde{B}$ its field strength. The duality at the level of the field strengths reads $\tilde{H}=\star H$ and $H=(-1)^{(p+1)(q+1)}s\star \tilde{H}$ where $s$ is the sign of the determinant of the metric. The gauge parameters are, respectively, $\epsilon$ and $\tilde{\epsilon}$.\\
\subsection{Lorenz-like gauge fixing, residual gauge and asymptotic charges}\label{anasin}
Let us briefly discuss Lorenz-like gauge fixing in the context of $p$-form gauge theories. The equations of motion can be simplified by fixing a Lorenz-like gauge; hence the equations of motion reduce to
\begin{equation}
\Box B_{\mu_2...\mu_{p+1}}=0.
\label{eqmotred}
\end{equation}
This can be achieved by requiring \cite{manz2}
\begin{equation}
    \Box \epsilon_{\mu_3...\mu_{p+1}}=0, \quad \Box \epsilon_{\mu_4...\mu_{p+1}}=0, \quad \hdots \quad
    \Box \epsilon_{\mu_{p+1}}=0, \quad
    \Box \epsilon=0.
\end{equation}

The computation of asymptotic charges for $p$-form gauge fields has also been considered in \cite{Esmaeili:2020eua,Afshar:2018apx}. However, the authors focus on the critical dimension $D_c=2p+2$, in which the field strengths of the dual descriptions are forms of the same degree, namely $p+1$. The present work generalizes these results to arbitrary dimensions and partially reproduces their results in the special case $D_c=2p+2$ since the authors assume a different starting point. 

The asymptotic charge relevant for our discussion has its roots in the covariant phase space formalism à la Wald \cite{Lee:1990nz,Ciambelli:2022vot}; we consider the N\"other charge
\begin{equation}
    Q_{p,D}:=\lim_{r \rightarrow +\infty}\oint_{S^{D-2}_u} k^{ur}_{p,D} r^{D-2} d\Omega ,
    \label{leewald}
\end{equation}
where we already specialize to the codimension-2 celestial sphere at Minkowski null infinity. In \eqref{leewald}, $k^{\mu \nu}_{p,D}$ is the N\"other two-form\footnote{This is the generalization to the case of a generic $p$-form, computed in the Lee-Wald formalism, of the well-known N\"other two-form in electromagnetism.} for a $p$-form in dimension $D$, $S_u^{D-2}$ is the celestial sphere at Minkowski null infinity, and $r^{D-2}d\Omega$ is its integration measure. To compute the N\"other two-form $k^{\mu \nu}_{p,D}$ we first compute, by varying the lagrangian, the presymplectic potential
\begin{equation}
    \Theta^{\mu_1}=-\frac{p+1}{(p+1)!}H^{\mu_1...\mu_{p+1}}\delta B_{\mu_2...\mu_{p+1}}
\end{equation}
from which we get the N\"other two-form
\begin{equation}
    k^{\mu_1\mu_2}_{p,D}=-\frac{(p+1)p}{(p+1)!}\epsilon_{\mu_3...\mu_{p+1}}H^{\mu_{1}...\mu_{p+1}}=:\Lambda_p \epsilon_{\mu_3...\mu_{p+1}}H^{\mu_{1}...\mu_{p+1}},  \qquad p>0,
\end{equation}
which is antisymmetric in $\mu_1$ and $\mu_2$ due to the total antisymmetry of $H^{\mu_{1}...\mu_{p+1}}$. For $p=1$ this reduces to the well-known N\"other two-form of electromagnetism. We also note that it would be linear in the field variations, which makes the verification of the integrability condition straightforward. The expression of the N\"other charge in the Bondi patch is given explicitly by
\begin{equation}
\begin{aligned}
    Q_{p,D}=&\Lambda_p\lim_{r \rightarrow +\infty}\oint_{S^{D-2}_u} \epsilon_{\mu_3...\mu_{p+1}}H^{ur\mu_3...\mu_{p+1}}r^{D-2}d\Omega=\\
    =&\Lambda_p\lim_{r \rightarrow +\infty}\oint_{S^{D-2}_u} \epsilon_{i_1...i_{p-1}}H^{uri_1...i_{p-1}}r^{D-2}d\Omega=\\
    =&\Lambda_p\lim_{r \rightarrow +\infty}\oint_{S_u^{D-2}}\gamma^{i_1j_1}...\gamma^{i_{p-1}j_{p-1}}\epsilon_{i_1...i_{p-1}}H_{ru j_1...j_{p-1}}r^{D-2p}d\Omega,
    \label{caricapforma}
\end{aligned}    
\end{equation}
which can be recognized as an electric-like charge since it contains field strength components $H_{ru j_1...j_{p-1}}$, and we will refer to it as $Q^{\mathrm{(e)}}_{p,D}$. We are able to partially reproduce, in $D_c=2p+2$, the results of \cite{Esmaeili:2020eua,Afshar:2018apx}, since the authors assume a different starting point invoking a well-defined variational principle for the action. In this sense we reproduce the charge associated with their bulk action term. This could lead to non-conservation of charges, but adopting a more modern perspective this can be a feature rather than an issue. Indeed, one should rather associate to non-conserved charges the meaning of observables, which can evolve
in time and become conserved only in the vacuum of the theory. 

Moreover, the N\"other charge \eqref{caricapforma}, as well as its dual counterpart, which will be introduced in the next paragraph, can have different behaviors when we make the limit explicit. On the one hand, charges can be well-defined, meaning that in both dual descriptions the charges are finite, non-vanishing and have the same radial behavior in the limit $r \rightarrow +\infty$. On the other hand, charges can be power-law divergent or vanishing; those charges diverge or vanish in the limit $r \rightarrow +\infty$ with a specific polynomial degree of divergence or vanishing. The discussion about duality and the information this duality gives on the charges strongly depends on what type of charge we are dealing with. For this reason, the choice of fall-offs is a fundamental part of the discussion: we focus on Coulomb and radiation fall-offs, which reflect two physically distinct regimes of the field in the limit $r \rightarrow +\infty$ of the Bondi patch of Minkowski spacetime. The Coulomb fall-off arises, for example, from static sources and its behaviour is non-radiative and time-independent, encoding only the presence and distribution of charges (or extended charged objects in the case of $p$-form fields). In contrast, the radiation fall-off is produced by genuinely time-dependent configurations, such as accelerating charges or propagating waves. It corresponds to the flux of energy carried by outgoing radiation and hence to soft radiation at $\mathcal{I}^+$. Physically, the Coulomb term captures the “memory” of the charge distribution, whereas the radiation term describes the dynamical degrees of freedom that reach null infinity of Minkowski spacetime.

With this in mind, we assume a polyhomogeneous expansion, i.e. an expansion containing also logarithmic terms, for $\epsilon_{i_1...i_{p-1}}$
\begin{equation}
    \epsilon_{i_1...i_{p-1}}=\sum_{l\in \frac{1}{2}\mathbb{Z}}\frac{\epsilon_{i_1...i_{p-1}}^{(l)}(u,\{x^i\})+\bar{\epsilon}_{i_1...i_{p-1}}^{(l)}(u,\{x^i\})ln(r)}{r^l}.
        \label{expansionformflin}
\end{equation}
However, to compute the leading term of the expansion, it is not necessary to perform the full computation for a generic $p$-form: the cases $p=1$ and $p=2$ are enough. Indeed, following the $p=1$ and $p=2$ cases, we can ask what the leading order of the gauge parameter components expansion \eqref{expansionformflin} is that preserves the Lorenz-like gauge and the radiation or Coulomb fall-offs of the fields. Let us call this order $X$: so we would have
\begin{equation}
\epsilon_{i_1...i_{p-1}}=\frac{\epsilon_{i_1...i_{p-1}}^{(X)}(\{x^i\})}{r^X}+\sum_{l> X}\frac{\epsilon_{i_1...i_{p-1}}^{(l)}(u,\{x^i\})+\bar{\epsilon}_{i_1...i_{p-1}}^{(l)}(u,\{x^i\})ln(r)}{r^l},
\end{equation}
where $X:=f(p,D)$ is a function of the dimension $D$ and the form degree $p$. The derivatives acting on the gauge parameter to get the gauge transformation do not change the type of function $f(p,D)$ is; therefore, since we have assumed radiation or Coulomb fall-offs, which have a would-be $f(p,D)$ linear both in $D$ and in $p$, we can conclude that $f(p,D)$ for the gauge parameter components \eqref{expansionformflin} must be linear both in $D$ and in $p$. We know its value for $p=1$ and $p=2$ (see \cite{manz2} for the full computation of the $p=1$ and $p=2$ cases); therefore, for radiation fall-offs we have
\begin{equation}
    f(p,D)=a(D)p+b(D), \qquad  \textnormal{with} \quad f(1,D)=\frac{D-4}{2}, \quad f(2,D)=\frac{D-6}{2}, 
\end{equation}
while for Coulomb fall-offs we have
\begin{equation}
\begin{aligned}
f(p,D)=a'(D)p+b'(D),\qquad  \textnormal{with} \quad  f(1,D)=D-4, \quad f(2,D)=D-6.\\
\end{aligned}    
\end{equation}
From the above relations we get $a(D)=-1$ and $b(D)=\frac{D-2} {2}$ for radiation fall-offs, while  $a'(D)=-2$ and $b'(D)=D-2$ for Coulomb fall-offs. Therefore we have
\begin{equation}
\begin{aligned}
    &X_{\text{rad}}:=f(p,D)=-p+\frac{D-2}{2}=\frac{D-(2p+2)}{2} \qquad \textnormal{radiation fall-off};\\
    &X_{\text{Cou}}:=f(p,D)=-2p+D-2=D-(2p+2) \qquad \textnormal{Coulomb fall-off};
\end{aligned}    
\end{equation}
The above computation, for $p$-forms in Minkowski spacetime, gives us the leading-order terms of the asymptotic expansion of $\epsilon_{i_1...i_{p-1}}$.
The asymptotic charge \eqref{caricapforma} has different asymptotic behaviors depending on which fall-offs are chosen for the fields.

For radiation fall-offs, the charge is given by  
\begin{equation}
\begin{aligned}
    &Q^{\mathrm{(e)}}_{p,D} \sim \Lambda_p\oint_{S_u^{D-2}} \gamma^{i_1j_1}...\gamma^{i_{p-1}j_{p-1}} \epsilon_{i_1...i_{p-1}}^{(\frac{D-(2p+2)}{2})}(\{x^i\}) \mathcal{R}^{(p)}\mathcal{O}(r^0)d\Omega,
    \label{chargeprad}
\end{aligned}
\end{equation}
where 
\begin{equation}
\begin{aligned}
 &\mathcal{R}^{(p)}:=H_{ruj_1...j_{p-1}}^{(\frac{D-2p+2}{2})};\\
 &H_{ruj_1...j_{p-1}}^{(\frac{D-2p+2}{2})}=-K_{\text{Rad}}B_{uj_1...j_{p-1}}^{(\frac{D-2p}{2})}-\sum_{k=1}^{p-1}\partial_{j_k}B^{(\frac{D-2p+2}{2})}_{u...j_{k-1}rj_{k+1}...j_{p-1}}-\partial_uB_{rj_1...j_{p-1}}^{(\frac{D-2p+2}{2})}+\bar{B}_{uj_1..j_{p-1}}^{(\frac{D-2p}{2})};\\
 &K_{\text{rad}}:=\frac{D-2p}{2}.
\end{aligned} 
\end{equation}
The asymptotic charge \eqref{chargeprad} is well-defined for every $p$ in every $D$ and it is formally independent of $p$: if we perform the substitution $p \mapsto q$ we obtain again a finite non-vanishing charge. The term $\bar{B}_{uj_1...j_{p-1}}^{(\frac{D-2p}{2})}$ is present only if we start with logarithmic terms in the expansion of the field components.

For the case of Coulomb fall-offs we get
\begin{equation}
\begin{aligned}
    &Q^{\mathrm{(e)}}_{p,D} \sim \Lambda_p\oint_{S_u^{D-2}} \gamma^{i_1j_1}...\gamma^{i_{p-1}j_{p-1}} \epsilon_{i_1...i_{p-1}}^{({D-(2p+2)})}(\{x^i\})\bigg[\mathcal{C}^{(p)}+\bar{\mathcal{C}}^{(p)}ln(r)\bigg]\mathcal{O}(r^{-(D-2p-2)})d\Omega, 
    \label{chargepcoul}
\end{aligned}
\end{equation}
where
\begin{equation}
\begin{aligned}
   &\mathcal{C}^{(p)}:=H_{ruj_1...j_{p-1}}^{(D-2p)}, \quad \bar{\mathcal{C}}^{(p)}:=\bar{H}_{ruj_1...j_{p-1}}^{(D-2p)};\\
   &H_{ruj_1...j_{p-1}}^{(D-2p)}=-K_{\text{Cou}}B_{uj_1...j_{p-1}}^{(D-2p-1)}-\sum_{k=1}^{p-1}\partial_{j_k}B^{(D-2p)}_{u...j_{k-1}rj_{k+1}...j_{p-1}}-\partial_uB_{rj_1...j_{p-1}}^{(D-2p)}+\bar{B}_{uj_1..j_{p-1}}^{(D-2p-1)};\\
   &\bar{H}_{ruj_1...j_{p-1}}^{(D-2p)}=-K_{\text{Cou}}\bar{B}_{uj_1...j_{p-1}}^{(D-2p-1)}-\partial_u\bar{B}_{rj_1...j_{p-1}}^{(D-2p)}; \\
   &K_{\text{Cou}}:=D-2p-1.
\end{aligned}   
\end{equation}
The logarithmic term is present only if we admit, from the beginning, a logarithmic leading order in the Coulomb fall-off expansion of the field components $B_{uj_1...j_{p-1}}$. Moreover, the logarithmic term is vanishing in the critical dimension $D_c=2p+2$ since
\begin{equation}
    \bar{H}_{ruj_1...j_{p-1}}^{(D-2p)}=(D-2p-2)\bar{B}_{uj_1...j_{p-1}}^{(D-2p-1)};
\end{equation}
and the charge \eqref{chargepcoul} has the same leading order as the charge \eqref{chargeprad}, as it should be\footnote{This general result is in agreement with what is known in the literature \cite{Afshar:2018apx}, but it is derived here in a different way.}. We note that, apart from the logarithmic term, which can be avoided by considering it as a pure gauge term, this charge is power-law vanishing in $D>2p+2$ while it has a power-law divergence in $D<2p+2$. Note that the charge obtained by the formal substitution $p \mapsto q$, which is the analogous electric-like charge but in the theory of the $q$-form, is power-law vanishing in $D<2p+2$ while it has a power-law divergence in $D>2p+2$.

The topic of divergent charges opens up questions about their renormalization; power-law divergent charges could be renormalizable with procedures similar to those used in \cite{Romoli:2024hlc,Manzoni:2025wyr}. The point is that counterterms can be added to the bulk action (or directly to the pre-symplectic
potential) to cure these divergences. In the covariant
phase space formalism à la Wald, the definition of charges comes from the Lagrangian. The former
are codimension-2 forms while the latter is a top form. It is therefore not surprising that the
notion of charges comes with built-in ambiguities and it is the choice of these ambiguities that allows us to properly renormalize divergent charges. However, let us wait for Paragraph \ref{dualmap2}  to go into more detail on renormalization topics.\\

\subsection{Magnetic-like charges for $p$-form gauge theories}\label{dualitymap}
We focus on the case of well-defined charges, in which the charges of both dual descriptions are finite and non-vanishing in the $r \rightarrow +\infty$ limit: there is no restriction in the case of radiation fall-offs, while we need to restrict to $D_c=2p+2$ in the case of Coulomb fall-offs\footnote{Indeed, if $D=2p+2$ we have that $D-2p-2=0$ and $D-2q-2=D-2(D-p-2)-2=-D+2p+2=0$.}. Given the electric-like charge of a $p$-form gauge theory \eqref{caricapforma} we can construct, by using the Hodge duality equations, a magnetic-like charge that contains the field strength components $\tilde{H}_{j_1...j_{q+1}}^{(\frac{D-(2q+2)}{2})}$. Let us consider the electric-like charge involving radiation fall-offs \eqref{chargeprad} for both the $p$-form and the $q$-form theory:
\begin{equation} \label{Qpq2D}
\begin{aligned}
    &Q^{\mathrm{(e)}}_{p,D} \sim \Lambda_p\oint_{S_u^{D-2}} \gamma^{i_1j_1}\cdots \gamma^{i_{p-1}j_{p-1}} \epsilon_{i_1\cdots i_{p-1}}^{(\frac{D-(2p+2)}{2})}(\{x^i\}) H_{ruj_1\cdots j_{p-1}}^{(\frac{D-2p+2}{2})} )d\Omega; \\
    &Q^{\mathrm{(e)}}_{q,D} \sim \Lambda_q\oint_{S_u^{D-2}} \gamma^{i_1j_1}\cdots \gamma^{i_{q-1}j_{q-1}} \tilde{\epsilon}_{i_1\cdots i_{q-1}}^{(\frac{D-(2q+2)}{2})}(\{x^i\}) \tilde{H}_{ruj_1\cdots j_{q-1}}^{(\frac{D-2q+2}{2})}d\Omega.
\end{aligned}
\end{equation}
From the Hodge duality equations between $H$ and $\tilde{H}$ we get, in particular,
\begin{equation} \label{dualpqD}
\begin{aligned}
    &\frac{H_{i_1\cdots i_{p+1}}}{r^{D-2-2(q-1)}} =(-1)^{(p+1)(q+1)}s \gamma^{k_1j_1}\cdots \gamma^{k_{q-1}j_{q-1}}\varepsilon_{urk_1\cdots k_{q-1}i_1\cdots i_{p+1}}\tilde{H}_{ruj_1\cdots j_{q-1}}\sqrt{\gamma};\\
    &\frac{\tilde{H}_{i_1\cdots i_{q+1}}}{r^{D-2-2(p-1)}} =\gamma^{k_1j_1}\cdots \gamma^{k_{p-1}j_{p-1}}\varepsilon_{urk_1\cdots k_{p-1}i_1\cdots i_{q+1}}H_{ruj_1\cdots j_{p-1}}\sqrt{\gamma}.
\end{aligned}
\end{equation}
where $s$ is the of the determinant of the metric tensor matrix representation. Moreover, looking at the Hodge dual equation with one index $u$ and the remaining ones angular indices gives us
\begin{equation}
    H_{ui_1\cdots i_p}^{(\frac{D-2p-2}{2})}=(-1)^{(p+1)(q+1)}s\varepsilon_{uri_1\cdots i_pk_1\cdots k_q}\gamma^{j_1k_1}\cdots \gamma^{j_qk_q}\tilde{H}_{uj_1\cdots j_q}^{(\frac{D-2q-2}{2})}\sqrt{\gamma},
\end{equation}
which implies a relation between the asymptotic electric and magnetic vector potentials, up to an angle-dependent integration constant which can be fixed in such a way that $B_{i_1\cdots i_p}^{(\frac{D-2p-2}{2})}$ and 
$\tilde{B}_{i_1\cdots i_q}^{(\frac{D-2q-2}{2})}$ are proportional. 

We would like to employ \eqref{dualpqD} to trade the integrands in \eqref{Qpq2D} for angular components of the corresponding dual field strengths. Making use of this relation we can map the electric-like charge for the $p$-form to a magnetic-like charge for the $q$-form\footnote{We note that by introducing the appropriate volume form on the sphere, \eqref{hodgepq} can be rewritten in a manifestly covariant form.}:
\begin{equation} \label{hodgepq}
\begin{aligned}
Q^{\mathrm{(e)}}_{p,D} \mapsto \tilde{Q}^{\mathrm{(m)}}_{q,D} \sim  \Lambda_p\oint_{S_u^{D-2}}\frac{\varepsilon^{i_1\cdots i_{p-1}j_1\cdots j_{q+1}}{\epsilon}_{i_1\cdots i_{p-1}}^{(\frac{D-(2p+2)}{2})}(\{x^i\})}{\sqrt{\gamma}}\tilde{H}_{j_1\cdots j_{q+1}}^{(\frac{D-(2q+2)}{2})} d\Omega,
\end{aligned}    
\end{equation}
To obtain the magnetic-like charge $\tilde{Q}^{\mathrm{(m)}}_{p,D}$, associated to $Q^{\mathrm{(e)}}_{q,D}$, it is enough to exchange $p \leftrightarrow q$ in \eqref{hodgepq} and taken into account the right numerical factor due to the Hodge operator.

Between the electric-like and magnetic-like charges there is a generalization of the relations holding for $D=4$ electromagnetism
\begin{equation}
    Q^\mathrm{(e)}_{p,D}=\tilde{Q}^\mathrm{(m)}_{q,D}, \ \ \ Q^\mathrm{(e)}_{q,D}=(-1)^{(p+1)(q+1)}s\tilde{Q}^\mathrm{(m)}_{p,D}.
    \label{eqelecmagn}
\end{equation}
where $s$ is the sign of the determinant of the metric tensor matrix representation. Moreover,
following \cite{Strominger_2016}, we can merge electric-like and magnetic-like asymptotic charges into an electromagnetic-like one
\begin{equation}
\begin{aligned}
&\mathcal{Q}^{\mathrm{(em)}}_{p,D}:=Q^\mathrm{(e)}_{p,D}+i\tilde{Q}^\mathrm{(m)}_{p,D}, \quad \mathcal{Q}^\mathrm{(em)}_{q,D}:=Q^\mathrm{(e)}_{q,D}+i\tilde{Q}^\mathrm{(m)}_{q,D},
\label{electromagncharge}
\end{aligned}    
\end{equation}
which, as in the case of electromagnetism in $D=4$, generate a complexified $\mathfrak{u}_{\mathbb{C}}(1)$ symmetry algebra; the first one in the theory described by the $p$-form while the second one in the theory described by the dual $q$-form. Under duality these charges transform according to M\"obius transformations which can be parametrized by the matrix $A \in \mathrm{PGL}(2,\mathbb{C})$ given by
\begin{equation}
    A=\begin{bmatrix}
0 & 1 \\
(-1)^{(p+1)(q+1)}s & 0  \\
\end{bmatrix};
\label{apgl}
\end{equation}
indeed, since $\mathbb{C} \cong \mathbb{R}^2$, we have
\begin{equation}
\begin{aligned}
    &\mathcal{Q}^\mathrm{(em)}_{p,D} \cong \begin{bmatrix}
Q_{p,D}^\mathrm{(e)}  \\
\tilde{Q}_{p,D}^\mathrm{(m)}  \\
\end{bmatrix}, \ \ \  \mathcal{Q}^\mathrm{(em)}_{q,D} \cong \begin{bmatrix}
Q_{q,D}^\mathrm{(e)}  \\
\tilde{Q}_{q,D}^\mathrm{(m)}  \\
\end{bmatrix}; \\
&\tilde{\mathcal{Q}}^\mathrm{(em)}_{q,D} \cong \begin{bmatrix}
\tilde{Q}_{q,D}^\mathrm{(m)}  \\
(-1)^{(p+1)(q+1)}sQ_{q,D}^{\mathrm{(e)}} \\
\end{bmatrix}, \ \ \  \tilde{\mathcal{Q}}^{\mathrm{(em)}}_{p,D} \cong \begin{bmatrix}
(-1)^{(p+1)(q+1)}s\tilde{Q}_{p,D}^{\mathrm{(m)}}  \\
Q_{p,D}^{\mathrm{(e)}}  \\
\end{bmatrix}, \ \ \ 
\end{aligned}
\end{equation}
and so
\begin{equation}
    \begin{bmatrix}
\tilde{Q}_{q,D}^{\mathrm{(m)}}  \\
(-1)^{(p+1)(q+1)}sQ_{q,D}^{\mathrm{(e)}}  \\
\end{bmatrix}=A
\begin{bmatrix}
 Q_{q,D}^{\mathrm{(e)}}  \\
\tilde{Q}_{q,D}^{\mathrm{(m)}}  \\
\end{bmatrix}, \ \ \ \ \ \begin{bmatrix}
(-1)^{(p+1)(q+1)}s\tilde{Q}_{p,D}^{\mathrm{(m)}}  \\
Q_{p,D}^{\mathrm{(e)}}  \\
\end{bmatrix}=A^{-1}
\begin{bmatrix}
 Q_{p,D}^{\mathrm{(e)}}  \\
\tilde{Q}_{p,D}^{\mathrm{(m)}}  \\
\end{bmatrix}.
\end{equation}
\begin{figure}[H]
    \centering  
    \includegraphics[width=9cm]{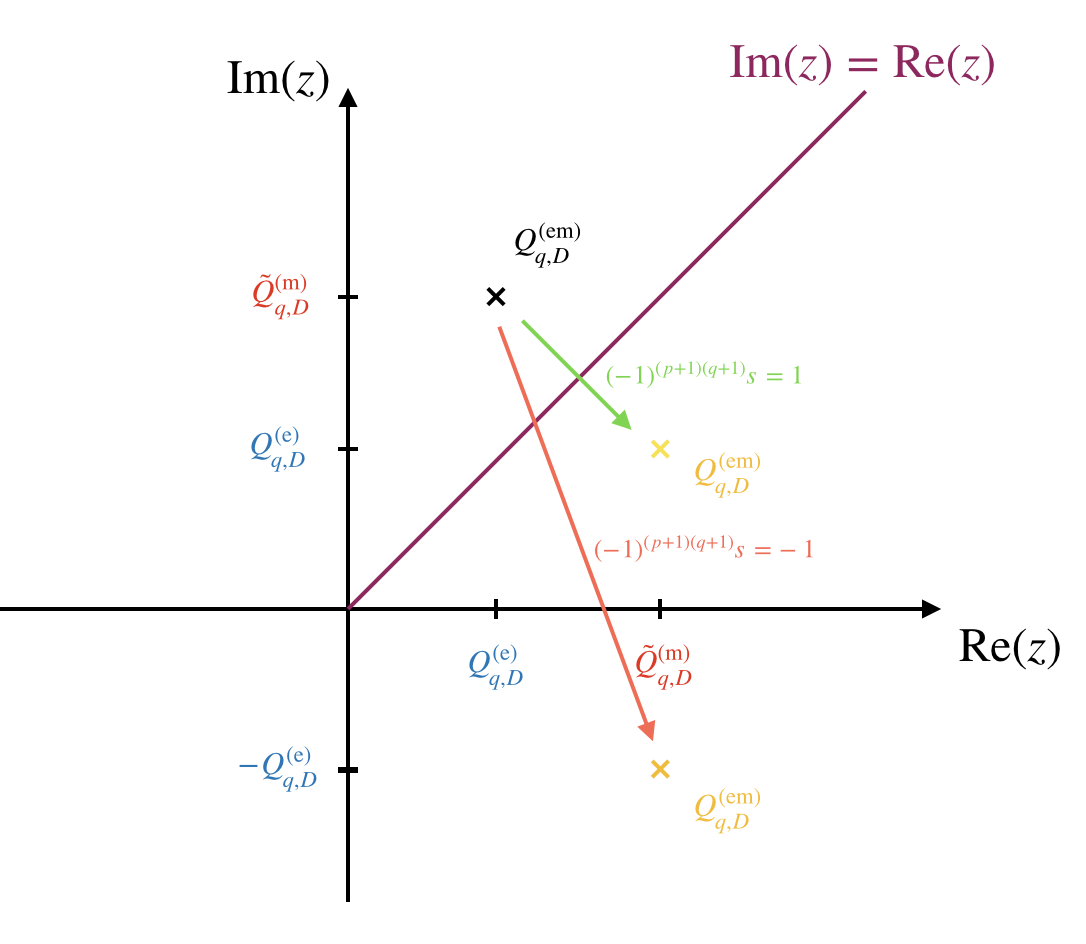}
    \caption{\textit{Schematic picture of the action of the duality on the complex electromagnetic-like asymptotic charge $\mathcal{Q}_{q,D}^{(\mathrm{em})}$. A similar scheme can be drawn for $\mathcal{Q}_{p,D}^{(\mathrm{em})}$.}}
    \label{dualityoncharges}
\end{figure}
According to the value of $(-1)^{(p+1)(q+1)}s$, the M\"obius transformation is a reflection with respect to the bisector of the first quadrant if $(-1)^{(p+1)(q+1)}s=1$, and a rotation of $\theta=\pm\frac{\pi}{2}$ (clockwise on $\mathcal{Q}^{\mathrm{(em)}}_{q,D}$ and counterclockwise on $\mathcal{Q}^{\mathrm{(em)}}_{p,D}$) if $(-1)^{(p+1)(q+1)}s=-1$. 
The M\"obius structure of the duality seems not to be related in an simple way to the lorentzian $\mathrm{SL}(2,\mathbb{C})$ acting on the celestial sphere. However, we can geometrize both action in a unified framework involving principal and equivariant bundles on the celestial sphere $S^2\cong \mathbb{CP}^1$.
\subsection{Geometrization of CCFT}\label{CCFT}

The above observation could be interesting for celestial holography and CCFTs. Let us recall that the group of M\"obius transformations in $D=4$ acts on the celestial sphere in several equivalent descriptions. In
fact, one has the natural chain of isomorphisms
\begin{equation}
    \mathrm{PGL}(2,\mathbb{C})
    \;\cong\;
    \mathrm{PSL}(2,\mathbb{C})
    \;\cong\;
    \mathrm{SO}^+(3,1) \, .
\end{equation}
The first isomorphism follows from the observation that any matrix
\(A \in \mathrm{GL}(2,\mathbb{C})\) can be rescaled to have determinant
\(1\) defining 
\begin{equation}
    \tilde A := (\det A)^{-1/2} A \in \mathrm{SL}(2,\mathbb{C}) ,
\end{equation}
where the factor \((\det A)^{-1/2}\) exists because every nonzero complex number admits a square root. Thus, by modding for scalar matrices,
\begin{equation}
    \mathrm{PGL}(2,\mathbb{C})
    \;\cong\;
    \mathrm{SL}(2,\mathbb{C}) / \{\pm I\}
    \;=\;
    \mathrm{PSL}(2,\mathbb{C}) \, .
\end{equation}
The second isomorphism follows form the realization that \(\mathrm{SL}(2,\mathbb{C})\) is the double covering of $\mathrm{SO}^+(3,1)$. Hence from the following short exact sequence 
\begin{equation}
    1 \;\rightarrow\; \{\pm I\}
    \;\rightarrow\; \mathrm{SL}(2,\mathbb{C})
    \;\rightarrow\; \mathrm{SO}^{+}(3,1)
    \;\rightarrow\; 1
    \label{seqesat}
\end{equation}
we get the desired isomorphism\footnote{The standard theoretical physics venue to present this isomorphism is by letting \(A \in \mathrm{SL}(2,\mathbb{C})\) act by
conjugation on \(2\times 2\) hermitian matrices:
\[
    X \;\longmapsto\; A X A^\dagger ,
\]
after identifying Minkowski vectors with hermitian matrices via
\[
    x^\mu \;\longleftrightarrow\;
    X = x^\mu \sigma_\mu , \qquad
    \sigma_\mu = (\mathbb{I}, \sigma_i) .
\] 
The kernel is precisely \(\{\pm I\}\), hence
\[
    \mathrm{PSL}(2,\mathbb{C})
    \;\cong\;
    \mathrm{SO}^+(3,1) \, .
\]}.
In celestial holography one often works directly with
\(\mathrm{SL}(2,\mathbb{C})\), since it provides the double covering of
\(\mathrm{SO}^+(3,1)\) and is convenient for keeping track of
representations with spin. Here we will follow a slightly different path. Indeed, as a first step, we will consider the proper orthochronous Lorentz group to construct our principal bundle, and we will then use its spin structure to lift it.

The geometrization proposal is to think the CCFT as a $\text{PGL}(2,\mathbb{C})$-principal M\"obius-equivariant bundle on the celestial sphere where the structure group captures the Lorentz symmetry (taking into consideration both tensors and spinors representations), while the M\"obius equivariance captures the duality. A principal \(\mathrm{PGL}(2,\mathbb{C})\)-bundle over the celestial sphere \(S^2 \cong \mathbb{CP}^1\) 
is a quadruple \((P,\pi,S^2,\Phi_{\bullet g})\) consisting of a total space \(P\), a projection
\begin{equation}
  \pi : P \rightarrow S^2 ,
\end{equation}
and a right action
\begin{equation}
  \Phi_{\bullet g} : P \times \mathrm{PGL}(2,\mathbb{C}) \rightarrow P ,
  \qquad (p,g) \mapsto \Phi_{\bullet g}(p)=p\bullet g ,
\end{equation}
which is free and transitive on fibers. In particular, the freeness and transitivity of the action imply that each 
fiber can be canonically identified with the group 
\(\mathrm{SL}(2,\mathbb{C})\) itself. Note that both the fiber and the base are $\mathrm{PGL(2,\mathbb{C})}$-spaces. If $g \in \mathrm{PGL(2,\mathbb{C})}$ is represented by
\begin{equation}
  M=\begin{bmatrix}
a & b \\
c & d  \\
\end{bmatrix}, \qquad ad-bc\neq 0;
\end{equation}
then the action on $S^2 \cong \mathbb{C}\mathbb{P}^1$ is given by \begin{equation}
\Phi_{g\cdot} : \mathrm{PGL}(2,\mathbb{C}) \times \mathbb{C}\mathbb{P}^1 \rightarrow \mathbb{C}\mathbb{P}^1 ,
  \qquad (g,z) \mapsto \Phi_{g \cdot}(z)=g \cdot z =\frac{az+b}{cz+d},
\end{equation}  
while the action on the fiber is simply given by matrix multiplication 
\begin{equation}
\Phi_{g}: \mathrm{PGL(2,\mathbb{C})} \times \mathrm{PGL(2,\mathbb{C})} \rightarrow \mathrm{PGL(2,\mathbb{C})},
  \qquad (g,v) \mapsto \Phi_{g}(v)=gv.
\end{equation}  
The bundle is supposed to be locally trivial, i.e. there exist an open covering $\{ U_i \}$ of $S^2\cong \mathbb{CP}^1$ and local trivializations
\begin{equation}
\varphi_i : \pi^{-1}(U_i) \longrightarrow U_i \times  \text{SL}(2,\mathbb{C}) ,
\end{equation}
such that, on the overlaps $U_i \cap U_j$,the relations
\begin{equation}
\varphi_j \circ \varphi_i^{-1}(x,f)
= \bigl( x ,\, t_{ij}(x)\cdot f \bigr) \, ,
\qquad x \in U_i \cap U_j \, ,
\end{equation}
where
\begin{equation}
t_{ij} : U_i \cap U_j \longrightarrow \mathrm{PGL(2,\mathbb{C})},
\end{equation}
hold. The right action can be defined as 
        \begin{equation}
            \Phi_{\bullet g}(z,v)=(z,v)\bullet g :=(z,vg)
            \label{ades}
        \end{equation}
where $(z,v) \in P \cong_{\text{loc}} U \times \text{PGL}(2,\mathbb{C})$ where $U$ is an open neighborhood of $z$. We can also show that the bundle is equivariant, i.e. $\pi \circ \Phi_{g*}=\Phi_{g \cdot} \circ \pi$, with respet to the $\text{PGL}(2,\mathbb{C})$ action defined on the total space as
\begin{equation}
 \Phi_{g*} : \mathrm{PGL}(2,\mathbb{C}) \times P \rightarrow P ,
  \qquad (g,(z,v)) \mapsto  \Phi_{g*}(z,v)=g * (z,v):= (g \cdot z, gv).
   \label{asin}
\end{equation}
Indeed, 
\begin{equation}
    g \cdot z = \Phi_{g\cdot} \circ \pi(z,v), \qquad  \pi \circ ( \Phi_{g*}(z,v))=\pi(g \cdot z, gv)= g \cdot z.
\end{equation}
Moreover the two actions on $P$ commute, i.e. $\Phi_{g*} \circ \Phi_{h \bullet}=\Phi_{h\bullet} \circ \Phi_{g *}$, since 
\begin{equation}
\Phi_{g*} \circ \Phi_{h\bullet}(z,v)=\Phi_{g*} ( (z,v)\bullet h )
= g * (z,vh)
= ( g\cdot z ,\, g(vh) ) ,
\end{equation}
and
\begin{equation}
\Phi_{g\bullet} \circ \Phi_{g*}(z,v)=\Phi_{g\bullet}( g * (z,v) )
= (g\cdot z ,\, gv)\bullet h
= ( g\cdot z ,\, (gv)h ) .
\end{equation}
and by associativity of matrix multiplication $g(vh)=(gv)h$.

In summary, the relevant object could be a principal 
$\mathrm{PGL}(2,\mathbb{C})$-principal bundle $\pi : P \to S^{2}$ on the celestial sphere (right action \eqref{ades}) endowed with a second left action, given by \eqref{asin}, of the M\"obius group which makes it equivariant. These two actions commute, so $P$ becomes an equivariant principal bundle: the bundle structure is compatible with the symmetry of the system (i.e. the duality). This structure is interesting because celestial operators in the CCFT can be thought as sections of associated bundles. Let us call $\mathcal{A}_{(h,\bar{h})}^{\text{CCFT}}$ the CCFT algebra of operators with fixed $(h,\bar{h})$ and let us taken into account the tensorial representation 
\begin{equation}
\rho_{(h,\bar{h})}: \text{PGL}(2,\mathbb{C}) \rightarrow \text{Aut}(\mathcal{A}_{(h,\bar{h})}^{\text{CCFT}})
\end{equation}
which implements the transformation law 
\begin{equation}
O_{h,\bar h}(z',\bar z') 
= \left( \frac{\partial z'}{\partial z} \right)^{-h}
  \left( \frac{\partial \bar z'}{\partial \bar z} \right)^{-\bar h}
  O_{h,\bar h}(z,\bar z) \, .
\end{equation}
under conformal transformation of the celestial sphere base $S^2$. This is exactly how a primary operator with weights $(h,\bar{h})$ would transform. At this point we built up the associated bundle $P_{h,\bar h} := P \times_{ \rho_{(h,\bar h)}}(\mathcal{A}_{(h,\bar{h})}^{\text{CCFT}})$; the fibres of $E_{h,\bar h}$ are copies of the operator algebra and a section section of $E_{h,\bar h}$ is, in practice, a celestial operator field that assigns to each point of the celestial sphere (each direction 
$(z,\bar z)$) an operator that transforms under a conformal transformation with the appropriate weight $(h,\bar h)$. Indeed a section 
\(s : S^2 \to E_{h,\bar h}\) 
chooses, for each point \(z,\bar{z}\), an element of the fibre \(\mathcal{A}_{(h,\bar{h})}^{\text{CCFT}}\): a section is a function that assigns to every direction on the celestial sphere a celestial operator. But not just any function: it must satisfy the correct transformation law. In local coordinates, if \(U \subset S^2\) is an open set 
with a trivialization of the principal bundle, then the section is represented by a function
\begin{equation}  
s_U : U \to \mathcal{A}_{(h,\bar{h})}^{\text{CCFT}}.
\end{equation}
On the overlap \(U \cap V\), the expressions in two trivializations satisfy
\begin{equation}
s_V(z)=\rho_{(h,\bar h)}(t_{UV}(z))\, s_U(z),
\end{equation}
where \(t_{UV}(z)\in \mathrm{PGL}(2,\mathbb{C})\) is the conformal transition function. Writing the representation \(\rho_{(h,\bar h)}\) explicitly, this becomes
\begin{equation}
s_V(z)
= 
\left( \frac{\partial z'}{\partial z} \right)^{-h}
\left( \frac{\partial \bar z'}{\partial \bar z} \right)^{-\bar h}
\, s_U(z),
\end{equation}
which is the usual transformation law of primary fields. However, if we leave things as they are, we would only be able to accommodate tensor representations and not spinorial ones. To solve this problem we need a spin structure which allows us to perform a bundle lift from the $\mathrm{PGL}(2,\mathbb{C})$-principal bundle $P$ to the $\mathrm{SL}(2,\mathbb{C})$-principal bundle $P_{\mathrm{Spin}}$, i.e. the spin bundle. The spin bundle comes together with a bundle map
\begin{equation}  
\Phi : P_{\mathrm{Spin}} \rightarrow P
\end{equation}
that is fibrewise a \(2\!:\!1\) covering and is equivariant with respect to the covering group
homomorphism \(
\Phi_{\text{group}} : \mathrm{SL}(2,\mathbb{C}) \rightarrow \mathrm{PGL}(2,\mathbb{C})\), i.e.
\begin{equation}  
\Phi(u \bullet g) = \Phi(u)\bullet \Phi_{\text{group}}(g)
\qquad \forall \  u \in P_{\mathrm{Spin}},\; g \in \mathrm{SL}(2,\mathbb{C}).
\end{equation}
In local trivializations over the same open covering \(\{U_i\}\) of the celestial sphere $S^2\cong \mathbb{CP}^1$, the bundle \(P_{\mathrm{Spin}}\) is defined by
transition functions
\begin{equation}  
t^{(\text{Spin})}_{ij} : U_i \cap U_j \rightarrow \mathrm{SL}(2,\mathbb{C})
\end{equation}
which satisfy the cocycle conditions and moreover lift the original cocycle
\begin{equation} 
\Phi_{\text{group}} \circ t^{(\text{Spin})}_{ij} = t_{ij}.
\end{equation}
Thus the \(\mathrm{PGL}(2,\mathbb{C})\)-valued transformations, i.e. the conformal transformations, on the base are lifted to
\(\mathrm{SL}(2,\mathbb{C})\)-valued trasformations.
Let 
\begin{equation} 
\rho_{(h,\bar{h})} : \mathrm{SL}(2,\mathbb{C}) \to \text{Aut}(\mathcal{A}_{(h,\bar{h})}^{\text{CCFT}})
\end{equation}
be a representation (also a spinorial one), the associated bundle
$P^{(\text{Spin})}_{h,\bar h} := P_{\text{Spin}} \times_{ \rho_{(h,\bar h)}}(\mathcal{A}_{(h,\bar{h})}^{\text{CCFT}})$
has sections which are now celestial operators that can transform like a spinors. Some comments are in order
\begin{itemize}
    \item the geometry of the base manifold still uses the transition functions with values in the M\"obius group, implementing conformal transformations;
    \item the spin bundle \(P_{\mathrm{Spin}}\) and its associated bundles encode both spinorial and tensorial representations. However, only the spinorial part (those representations for which the non-trivial element of the kernel covering map $\Phi_{\text{group}}$ acts non-trivially) is genuinely new;
    \item spinor fields (sections of bundles associated to \(P_{\mathrm{Spin}}\))
    glue through the lifted transition functions and transform under the spinorial representations.
\end{itemize}
However, the existence of such a lift is topologically obstructed by some cohomology classes, hence we now briefly discuss why in our case the lift is possible. The short exact sequence \eqref{seqesat} induces a long exact sequence in cohomology and in particular induces the obstruction morphism\footnote{For the reader who is not comfortable with these topics, some useful references could be \cite{algetop,Griffiths:433962,Hatcher:478079}.} 
\begin{equation}
\delta : \check{H}^{1}\!\left(X,\operatorname{PGL}(2,\mathbb{C})\right)
\;\longrightarrow\;
H^{2}\!\left(X,\mathbb{Z}_2\right);
\end{equation}
where $\check{H}^{1}\!\left(X,\operatorname{PGL}(2,\mathbb{C})\right)$, the non-abelian Cech cohomology, classifies the possible $\mathrm{PGL}(2,\mathbb{C})$-principal bundle over $X$. The element $\delta([P])$, where $[P] \in H^{1}\!\left(X,\operatorname{PGL}(2,\mathbb{C})\right)$ is the obstruction to lift to the spin bundle. In our case the base is the celestial sphere $S^2 \cong \mathbb{CP}^1$ and we have
\begin{equation}
H^{2}\!\left(\mathbb{CP}^{1},\mathbb{Z}_2\right)\cong \mathbb{Z}_2 .
\end{equation}
Moreover, both $\operatorname{PGL}(2,\mathbb{C})$ and $\operatorname{SL}(2,\mathbb{C})$ are connected and so their principal bundles can be classified by their fundamental groups
\begin{equation}
\pi_{1}\!\left(\operatorname{PGL}(2,\mathbb{C})\right)
\cong \mathbb{Z}_2, \qquad \pi_{1}\!\left(\operatorname{SL}(2,\mathbb{C})\right)= 0.
\end{equation}
Therefore, there exist two isomorphism classes of $\operatorname{PGL}(2,\mathbb{C})$–principal bundles 
over $\mathbb{CP}^{1}$ and only one isomorphism class of $\operatorname{SL}(2,\mathbb{C})$–principal bundle.
The covering map induces
\begin{equation}
\pi_{1}\!\left(\operatorname{SL}(2,\mathbb{C})\right)
\;\longrightarrow\;
\pi_{1}\!\left(\operatorname{PGL}(2,\mathbb{C})\right) ,
\end{equation}
so the image is only the zero element. This means that a principal $\operatorname{PGL}(2,\mathbb{C})$-bundle over $\mathbb{CP}^{1}$ which is topologically trivial\footnote{Note that topologically trivial means that topologically the bundle is a global product but it can be non-trivial form the holomorphic point of view. In would be interesting to enter in the computation of the obstruction class in the case of topologically non-triviality but this will be left for more specific work.} can be lifted to an $\operatorname{SL}(2,\mathbb{C})$-bundle.

However, our principal bundle is also $\text{PGL}(2,\mathbb{C})$-equivariant and this property is inherited from associated bundles. In fact, on the associated bundles we can define the action 
\begin{equation}
    k \cdot (p,\mathcal{O}_{h,\bar{h}}) := (k * p, \sigma(k)\mathcal{O}_{h,\bar{h}}).
\end{equation}
where $k \in \text{K} \subset \text{PGL}(2,\mathbb{C})$. This action must be compatible with the structure of the associated bundle, i.e. the identification
\begin{equation}
    (p\bullet g, \mathcal{O}_{h,\bar{h}}) \sim_{\text{eq}} (p,\rho_{(h,\bar{h})}(g)\mathcal{O}_{h,\bar{h}})
\end{equation}
must translate into
\begin{equation}
    k \cdot (p\bullet g, \mathcal{O}_{h,\bar{h}}) \sim_{\text{eq}} k \cdot (p,\rho_{(h,\bar{h})}(g)\mathcal{O}_{h,\bar{h}}).
\end{equation}
This is assured if $\sigma(k)$ is a $\rho_{(h,\bar{h})}$-representation homomorphism.
The point is that we can use this $\text{K}$-equivariance to implement the duality at the level of the space of sections, i.e. the CCFT operators. For example on the associated bundle whose sections are CCFT operator currents associated to the asymptotic electromagnetic-like charges the duality would be implemented by the cyclic subgroup generated by $A$. Since $A^2 \propto I$, the cyclic subgroup is simply given by $\{I,A\} \cong \mathbb{Z}_2$. We note the condition to be an $\rho_{(h,\bar{h})}$-representation homomorphism means that the duality acts on the same conformal multiplet without change the representation $(h,\bar{h})$. This can be inconvenient, and one might want to require that duality also modifies the representation of the celestial operator. This can be achieved either by an external map between various associated bundles, or by more subtly defining the actions on the principal bundle in such a way that the action of the structure group and the equivariance group commutes up to an external automorphism $\alpha_k$ of the structure group itself. In this way, the action of the equivariance group (and hence of duality) can mix different representations of the structure group, automatically inducing a map between the various associated bundles. However, for $\text{PGL}(2,\mathbb{C})$, there are no non-trivial external automorphisms, and this rules out the second possibility. Similar consideration can be done for associated bundles of the spin bundle $P_{\text{Spin}}$ with the caveat that the $\text{PGL}(2,\mathbb{C})$-equivariance lift\footnote{This lifting uses the spin structure of the base $S^2 \cong \mathbb{CP}^1$ which exist and its unique since $\mathbb{CP}^1$ is a Riemann surface of genus $\text{g}=0$.} to a $\text{SL}(2,\mathbb{C})$-equivariance. Further investigations on this proposal should address the following issues:
\begin{enumerate}
    \item \textbf{Operator product expansions}. OPEs could be accomodate in the following geometrical picture since in a local trivialization of the principal bundle, the product of two primary fields 
$\mathcal{O}_{h_1,\bar{h}_1} \in \Gamma(P_{h_1,\bar{h}_1})$ and $\mathcal{O}_{h_1,\bar{h}_1} \in \Gamma(P_{h_2,\bar{h}_2})$ takes values in the 
tensor product bundle $P_{h_1,\bar{h}_1} \otimes P_{h_2,\bar{h}_2}$. If the representation $\rho_{(h_1,\bar{h}_1)\otimes(h_2,\bar{h}_2)}$ is reducible then the OPEs 
takes the form
\begin{equation}
\mathcal{O}_1(z_1,\bar z_1)\,\mathcal{O}_2(z_2,\bar z_2)
\;\sim\;
\sum_k C_{12}^k(z_{12},\bar z_{12})\,\mathcal{O}_k(z_2,\bar z_2) ,
\end{equation}
where the coefficients $C_{12}^k$ encode the projection of the local tensor product 
section onto the associated bundle $P_{h_k,\bar{h}_k}$.
    \item \textbf{Descendants}. Descendants arise by differentiating a primary section $O_{h,\bar h}\in\Gamma(P_{h,\bar h})$, but ordinary derivatives do not act covariantly on sections of an associated bundle. Thus one must use the covariant derivative induced by the connection on the principal $\mathrm{SL}(2,\mathbb{C})$-bundle whose non-trivial transformation could reproduce the characteristic mixing terms of descendants.  Since $(O,\nabla O,\nabla^2 O,\dots)$ cannot be considered sections of $P_{h,\bar h}$. Descendants are naturally accommodated as sections of the jet bundle. A jet bundle of finite order $J^k(P_{h,\bar h})$ contains only the data of a primary and its derivatives up to order $k$, so it can accommodate only finitely many descendants. However, the full Verma module generated by $O_{h,\bar h}$ contains all descendants, corresponding to arbitrary many derivatives.  Since no finite jet bundle can encode derivatives of all orders, the natural geometric home of all descendants is the infinite jet bundle $J^\infty(P_{h,\bar h})$, which captures the entire tower of conformal descendants in a single geometric object.
    \item \textbf{Total CCFT specturm}. The spectrum of the CCFT can be considered as the sheaf $\mathcal{A}$ such that 
\[
\mathcal A(U)
   = \widehat{\bigoplus_{(h,\bar h)}}\Gamma\!\left(U,J^\infty(P_{h,\bar h})\right) ,
\qquad U\subset S^2 .
\]
Since each infinite jet bundle is a smooth bundle, the
assignment $$U\mapsto\Gamma(U,J^\infty(P_{h,\bar h}))$$ is a sheaf of sections; moreover, direct (completed) sums of sheaves is a sheaf, hence $\mathcal A$ is effectively a sheaf.  In the limit $z_1 \rightarrow z_2$ the OPE defines a smooth local bundle morphism
\[
\mathrm{OPE}:\,
J^\infty(P_{h_1,\bar h_1})\otimes J^\infty(P_{h_2,\bar h_2})
\longrightarrow
\widehat{\bigoplus_k}J^\infty(P_{h_k,\bar h_k}),
\]
which is in turn compatible with locality; so OPE is a morphism of sheaves.  
Thus $\mathcal A$ is a sheaf of algebras where the OPE could be considered as the local
multiplication on this sheaf.
\end{enumerate}

\section{An existence and uniqueness theorem for the duality map}\label{secTHM}
We now want to consider the duality map
\begin{equation}
    Q^{(\mathrm{e})}_{q,D}\leftrightarrow Q^{(\mathrm{e})}_{p,D},
    \label{duality}
\end{equation} 
which links the electric-like charges of the two dual theories. Our goal is to prove an existence and uniqueness theorem for this duality map.

\subsection{Duality map for well-defined charges}
The existence and uniqueness of the duality map \eqref{duality} is stated by the following theorem.

\begin{Theorem}[Existence and uniqueness of the duality map for well-defined charges]\label{THM3.1} 
Let $(M_D,\boldsymbol{\eta})$ be the $D$-dimensional Minkowski spacetime where the de Rham cohomology groups in positive degree $n>0$ vanish, $H^n=0$. Let $p\in \big[1, \frac{D-2}{2}\big]$ and $q:=D-p-2$; let $H$ and $\tilde{H}$ be the field strengths of a $p$-form gauge field $B$ and of a $q$-form gauge field $\tilde{B}$, respectively, with fall-offs that ensure the property of well-defined asymptotic charges. Then there exists a duality map $f \in \mathrm{GL}(n_p,\mathbb{C})$ such that the following diagram
\begin{equation}
\centering
\begin{tikzcd}
\Omega_{\mathrm{AS}}^{p+1}(M_D) \arrow[rr, "\star"] & & \Omega_{\mathrm{AS}}^{q+1}(M_D)\\
\Omega_{\mathrm{AS}}^{p}(M_D) \arrow[u, "d_p"] \arrow[d, "\pi_1", swap ] \arrow[rr, "\star_{D-2}"] & & \Omega_{\mathrm{AS}}^{q}(M_D) \arrow[u, "d_q", swap] \arrow[d, "\pi_2"]\\
\mathbb{C}^{n_p}  \arrow[rr, "f"] & & \mathbb{C}^{n_{q}} \\
\end{tikzcd}
\end{equation}
commutes. Moreover, $f$ admits a unique restriction to a 1-dimensional subspace such that $f|_{Q}: Q^{(\mathrm{e})}_{p,D} \mapsto Q_{q,D}^{(\mathrm{e})}$ and $f^{-1}|_{Q}: Q^{(\mathrm{e})}_{q,D} \mapsto Q^{(\mathrm{e})}_{p,D}$.
\end{Theorem}

\begin{proof}
Let us consider two copies of the de Rham complex, one labelled by $p$, $C_{\text{dR}}^{*(p)}$, and one labelled by $q=D-p-2$, $C_{\text{dR}}^{*(q)}$:
\begin{equation}
\centering
\begin{tikzpicture}
[x=0.75pt,y=0.75pt,yscale=-1,xscale=1]

\draw    (252,175) -- (314.77,175.36) ;
\draw [shift={(316.77,175.37)}, rotate = 180.33] [color={rgb, 255:red, 0; green, 0; blue, 0 }  ][line width=0.75]    (10.93,-3.29) .. controls (6.95,-1.4) and (3.31,-0.3) .. (0,0) .. controls (3.31,0.3) and (6.95,1.4) .. (10.93,3.29)   ;
\draw    (342,175) -- (404.77,175.36) ;
\draw [shift={(406.77,175.37)}, rotate = 180.33] [color={rgb, 255:red, 0; green, 0; blue, 0 }  ][line width=0.75]    (10.93,-3.29) .. controls (6.95,-1.4) and (3.31,-0.3) .. (0,0) .. controls (3.31,0.3) and (6.95,1.4) .. (10.93,3.29)   ;
\draw    (449,176) -- (511.77,176.36) ;
\draw [shift={(513.77,176.37)}, rotate = 180.33] [color={rgb, 255:red, 0; green, 0; blue, 0 }  ][line width=0.75]    (10.93,-3.29) .. controls (6.95,-1.4) and (3.31,-0.3) .. (0,0) .. controls (3.31,0.3) and (6.95,1.4) .. (10.93,3.29)   ;
\draw    (150,175) -- (212.77,175.36) ;
\draw [shift={(214.77,175.37)}, rotate = 180.33] [color={rgb, 255:red, 0; green, 0; blue, 0 }  ][line width=0.75]    (10.93,-3.29) .. controls (6.95,-1.4) and (3.31,-0.3) .. (0,0) .. controls (3.31,0.3) and (6.95,1.4) .. (10.93,3.29)   ;
\draw    (252,115) -- (314.77,115.36) ;
\draw [shift={(316.77,115.37)}, rotate = 180.33] [color={rgb, 255:red, 0; green, 0; blue, 0 }  ][line width=0.75]    (10.93,-3.29) .. controls (6.95,-1.4) and (3.31,-0.3) .. (0,0) .. controls (3.31,0.3) and (6.95,1.4) .. (10.93,3.29)   ;
\draw    (342,115) -- (404.77,115.36) ;
\draw [shift={(406.77,115.37)}, rotate = 180.33] [color={rgb, 255:red, 0; green, 0; blue, 0 }  ][line width=0.75]    (10.93,-3.29) .. controls (6.95,-1.4) and (3.31,-0.3) .. (0,0) .. controls (3.31,0.3) and (6.95,1.4) .. (10.93,3.29)   ;
\draw    (449,116) -- (511.77,116.36) ;
\draw [shift={(513.77,116.37)}, rotate = 180.33] [color={rgb, 255:red, 0; green, 0; blue, 0 }  ][line width=0.75]    (10.93,-3.29) .. controls (6.95,-1.4) and (3.31,-0.3) .. (0,0) .. controls (3.31,0.3) and (6.95,1.4) .. (10.93,3.29)   ;
\draw    (150,115) -- (212.77,115.36) ;
\draw [shift={(214.77,115.37)}, rotate = 180.33] [color={rgb, 255:red, 0; green, 0; blue, 0 }  ][line width=0.75]    (10.93,-3.29) .. controls (6.95,-1.4) and (3.31,-0.3) .. (0,0) .. controls (3.31,0.3) and (6.95,1.4) .. (10.93,3.29)   ;
\draw    (124,125) -- (124.37,161.43) ;
\draw [shift={(124.39,163.43)}, rotate = 269.42] [color={rgb, 255:red, 0; green, 0; blue, 0 }  ][line width=0.75]    (10.93,-3.29) .. controls (6.95,-1.4) and (3.31,-0.3) .. (0,0) .. controls (3.31,0.3) and (6.95,1.4) .. (10.93,3.29)   ;
\draw    (530,125) -- (530.37,161.43) ;
\draw [shift={(530.39,163.43)}, rotate = 269.42] [color={rgb, 255:red, 0; green, 0; blue, 0 }  ][line width=0.75]    (10.93,-3.29) .. controls (6.95,-1.4) and (3.31,-0.3) .. (0,0) .. controls (3.31,0.3) and (6.95,1.4) .. (10.93,3.29)   ;
\draw    (420,125) -- (420.37,161.43) ;
\draw [shift={(420.39,163.43)}, rotate = 269.42] [color={rgb, 255:red, 0; green, 0; blue, 0 }  ][line width=0.75]    (10.93,-3.29) .. controls (6.95,-1.4) and (3.31,-0.3) .. (0,0) .. controls (3.31,0.3) and (6.95,1.4) .. (10.93,3.29)   ;
\draw    (327,125) -- (327.37,161.43) ;
\draw [shift={(327.39,163.43)}, rotate = 269.42] [color={rgb, 255:red, 0; green, 0; blue, 0 }  ][line width=0.75]    (10.93,-3.29) .. controls (6.95,-1.4) and (3.31,-0.3) .. (0,0) .. controls (3.31,0.3) and (6.95,1.4) .. (10.93,3.29)   ;
\draw    (224,125) -- (224.37,161.43) ;
\draw [shift={(224.39,163.43)}, rotate = 269.42] [color={rgb, 255:red, 0; green, 0; blue, 0 }  ][line width=0.75]    (10.93,-3.29) .. controls (6.95,-1.4) and (3.31,-0.3) .. (0,0) .. controls (3.31,0.3) and (6.95,1.4) .. (10.93,3.29)   ;

\draw (319,164.4) node [anchor=north west][inner sep=0.75pt]    {$\Omega^{q}$};
\draw (216,164.4) node [anchor=north west][inner sep=0.75pt]    {$\Omega^{q-1}$};
\draw (519,164.4) node [anchor=north west][inner sep=0.75pt]    {$\Omega^{q+2}$};
\draw (411,164.4) node [anchor=north west][inner sep=0.75pt]    {$\Omega^{q+1}$};
\draw (114,164.4) node [anchor=north west][inner sep=0.75pt]    {$\Omega^{q-2}$};
\draw (562,172.4) node [anchor=north west][inner sep=0.75pt]    {$...$};
\draw (97,172.4) node [anchor=north west][inner sep=0.75pt]    {$...$};
\draw (163,149.4) node [anchor=north west][inner sep=0.75pt]    {$d_{q-2}$};
\draw (466,149.4) node [anchor=north west][inner sep=0.75pt]    {$d_{q+1}$};
\draw (359,149.4) node [anchor=north west][inner sep=0.75pt]    {$d_{q}$};
\draw (267,149.4) node [anchor=north west][inner sep=0.75pt]    {$d_{q-1}$};
\draw (319,104.4) node [anchor=north west][inner sep=0.75pt]    {$\Omega^{p}$};
\draw (216,104.4) node [anchor=north west][inner sep=0.75pt]    {$\Omega^{p-1}$};
\draw (519,104.4) node [anchor=north west][inner sep=0.75pt]    {$\Omega^{p+2}$};
\draw (411,104.4) node [anchor=north west][inner sep=0.75pt]    {$\Omega^{p+1}$};
\draw (114,104.4) node [anchor=north west][inner sep=0.75pt]    {$\Omega^{p-2}$};
\draw (562,114.4) node [anchor=north west][inner sep=0.75pt]    {$...$};
\draw (97,114.4) node [anchor=north west][inner sep=0.75pt]    {$...$};
\draw (163,89.4) node [anchor=north west][inner sep=0.75pt]    {$d_{p-2}$};
\draw (466,89.4) node [anchor=north west][inner sep=0.75pt]    {$d_{p+1}$};
\draw (359,89.4) node [anchor=north west][inner sep=0.75pt]    {$d_{p}$};
\draw (267,89.4) node [anchor=north west][inner sep=0.75pt]    {$d_{p-1}$};
\draw (422,128.4) node [anchor=north west][inner sep=0.75pt]    {$\star _{D}$};
\draw (329,128.4) node [anchor=north west][inner sep=0.75pt]    {$\star _{D-2}$};
\draw (226,128.4) node [anchor=north west][inner sep=0.75pt]    {$\star _{D-4}$};
\draw (126,128.4) node [anchor=north west][inner sep=0.75pt]    {$\star _{D-6}$};
\draw (532,128.4) node [anchor=north west][inner sep=0.75pt]    {$\star _{D+2}$};
\end{tikzpicture}
\label{topocosa}
\end{equation}
where every $\star_{D+2n}$ with $n \in \mathbb{Z} \setminus \{-\infty,+\infty\}$ is required to be a group homomorphism and $\star_D$ is the Hodge operator. Noting that $q-p=D-p-2-p=D-2p-2$, we have that  
\begin{equation}
    \star^{*} : C_{\text{dR}}^{*(p)} \mapsto C_{\text{dR}}^{*(q)} \ \ \ \ [D-2p-2]
\end{equation}
is a homotopy of cochain complexes. In the critical dimension, i.e. $p=\frac{D-2}{2}=q$, the $p$-form gauge theory is self-dual and the homotopy $\star^*$ is an isomorphism of cochain complexes, since every space of forms is mapped to itself. 

Now, since we are interested in gauge field theories on Minkowski spacetime, we can assume trivial topology, i.e. all the cohomology groups ($n \neq 0$) of the de Rham complex are trivial:
\begin{equation}
    H^n=\frac{Z_n:=\{B\in \Omega^n | d_{n}B=0\}}{B_n:=\{d_{n-1}A \in \Omega^n | A \in \Omega^{n-1}\}}=0;
\end{equation}
this means that every cocycle is also a coboundary\footnote{In other words, every closed form is exact.}. Therefore, the de Rham complexes in \eqref{topocosa} are exact sequences of abelian groups. 

Now, let us restrict to only one de Rham complex\footnote{The same considerations hold for $C_{\text{dR}}^{*(q)}$.}, $C_{\textnormal{dR}}^{*(p)}$:
\begin{equation}
\centering
\tikzset{every picture/.style={line width=0.75pt}} 
\begin{tikzpicture}[x=0.75pt,y=0.75pt,yscale=-1,xscale=1]

\draw    (252,115) -- (314.77,115.36) ;
\draw [shift={(316.77,115.37)}, rotate = 180.33] [color={rgb, 255:red, 0; green, 0; blue, 0 }  ][line width=0.75]    (10.93,-3.29) .. controls (6.95,-1.4) and (3.31,-0.3) .. (0,0) .. controls (3.31,0.3) and (6.95,1.4) .. (10.93,3.29)   ;
\draw    (342,115) -- (404.77,115.36) ;
\draw [shift={(406.77,115.37)}, rotate = 180.33] [color={rgb, 255:red, 0; green, 0; blue, 0 }  ][line width=0.75]    (10.93,-3.29) .. controls (6.95,-1.4) and (3.31,-0.3) .. (0,0) .. controls (3.31,0.3) and (6.95,1.4) .. (10.93,3.29)   ;
\draw    (449,116) -- (511.77,116.36) ;
\draw [shift={(513.77,116.37)}, rotate = 180.33] [color={rgb, 255:red, 0; green, 0; blue, 0 }  ][line width=0.75]    (10.93,-3.29) .. controls (6.95,-1.4) and (3.31,-0.3) .. (0,0) .. controls (3.31,0.3) and (6.95,1.4) .. (10.93,3.29)   ;
\draw    (150,115) -- (212.77,115.36) ;
\draw [shift={(214.77,115.37)}, rotate = 180.33] [color={rgb, 255:red, 0; green, 0; blue, 0 }  ][line width=0.75]    (10.93,-3.29) .. controls (6.95,-1.4) and (3.31,-0.3) .. (0,0) .. controls (3.31,0.3) and (6.95,1.4) .. (10.93,3.29)   ;

\draw (319,104.4) node [anchor=north west][inner sep=0.75pt]    {$\Omega^{p}$};
\draw (216,104.4) node [anchor=north west][inner sep=0.75pt]    {$\Omega^{p-1}$};
\draw (519,104.4) node [anchor=north west][inner sep=0.75pt]    {$\Omega^{p+2}$};
\draw (411,104.4) node [anchor=north west][inner sep=0.75pt]    {$\Omega^{p+1}$};
\draw (114,104.4) node [anchor=north west][inner sep=0.75pt]    {$\Omega^{p-2}$};
\draw (562,114.4) node [anchor=north west][inner sep=0.75pt]    {$...$};
\draw (97,114.4) node [anchor=north west][inner sep=0.75pt]    {$...$};
\draw (163,89.4) node [anchor=north west][inner sep=0.75pt]    {$d_{p-2}$};
\draw (466,89.4) node [anchor=north west][inner sep=0.75pt]    {$d_{p+1}$};
\draw (359,89.4) node [anchor=north west][inner sep=0.75pt]    {$d_{p}$};
\draw (267,89.4) node [anchor=north west][inner sep=0.75pt]    {$d_{p-1}$};
\end{tikzpicture},
\end{equation}
and let us take into account the fact that we are interested in asymptotic symmetries. We start with a $p$-form gauge theory, with gauge field $B \in \Omega^p$ and field strength $H=d_pB \in \Omega^{p+1}$, which we require to be non-vanishing\footnote{Otherwise the asymptotic charge would be zero and would be associated with a trivial gauge transformation.}. Since $C_{\mathrm{dR}}^{*(p)}$ on Minkowski spacetime is exact, we have $H=0 \Leftrightarrow B=d_{p-1}A$ for some $A \in \Omega^{p-1}$; hence we need to discard all those elements $B \in \Omega^p$ such that $B=d_{p-1}A$ for some $A \in \Omega^{p-1}$. Moreover, only the zero form can have vanishing field strength, but again, by exactness, $B=0 \Leftrightarrow A=d_{p-2}C$ for some $C \in \Omega^{p-2}$. 

Therefore, for the purposes of asymptotic symmetries we can replace $\Omega^{p+1}$ with $\Omega^{p+1}_{\text{AS}}:=\{H \in \Omega^{p+1} | H=d_pB, H \neq 0\} \cup \{H=0\}$, while $\Omega^{p}$ remains the same but we denote it as $\Omega^{p}_{\text{AS}}$, and we replace $\Omega^{p-1}$ and $\Omega^{p+1}$ with $0$, thereby obtaining the asymptotic symmetries de Rham complex $C_{\text{ASdR}}^{*(p)}$:
\begin{equation}
\centering
\begin{tikzpicture}[x=0.75pt,y=0.75pt,yscale=-1,xscale=1]

\draw    (233,115) -- (295.77,115.36) ;
\draw [shift={(297.77,115.37)}, rotate = 180.33] [color={rgb, 255:red, 0; green, 0; blue, 0 }  ][line width=0.75]    (10.93,-3.29) .. controls (6.95,-1.4) and (3.31,-0.3) .. (0,0) .. controls (3.31,0.3) and (6.95,1.4) .. (10.93,3.29)   ;
\draw    (332,115) -- (394.77,115.36) ;
\draw [shift={(396.77,115.37)}, rotate = 180.33] [color={rgb, 255:red, 0; green, 0; blue, 0 }  ][line width=0.75]    (10.93,-3.29) .. controls (6.95,-1.4) and (3.31,-0.3) .. (0,0) .. controls (3.31,0.3) and (6.95,1.4) .. (10.93,3.29)   ;
\draw    (439,116) -- (501.77,116.36) ;
\draw [shift={(503.77,116.37)}, rotate = 180.33] [color={rgb, 255:red, 0; green, 0; blue, 0 }  ][line width=0.75]    (10.93,-3.29) .. controls (6.95,-1.4) and (3.31,-0.3) .. (0,0) .. controls (3.31,0.3) and (6.95,1.4) .. (10.93,3.29)   ;

\draw (300,104.4) node [anchor=north west][inner sep=0.75pt]    {$\Omega_{\text{AS}}^{p}$};
\draw (216,108.4) node [anchor=north west][inner sep=0.75pt]    {$0$};
\draw (509,108.4) node [anchor=north west][inner sep=0.75pt]    {$0$};
\draw (401,104.4) node [anchor=north west][inner sep=0.75pt]    {$\Omega_{\text{AS}}^{p+1}$};
\draw (456,89.4) node [anchor=north west][inner sep=0.75pt]    {$d_{p+1}$};
\draw (349,89.4) node [anchor=north west][inner sep=0.75pt]    {$d_{p}$};
\draw (248,89.4) node [anchor=north west][inner sep=0.75pt]    {$d_{p-1}$};
\end{tikzpicture}
\end{equation} 
which is a short exact sequence. By general arguments or by explicit computation it follows that $d_p$ is an isomorphism. Similarly, $d_q$ is an isomorphism. We emphasize that $d_p^{-1}$ and $d_q^{-1}$ are defined only on the image, respectively, of $d_p$ and $d_q$, and that their definition relies on the trivial de Rham cohomology of Minkowski spacetime.

Let us now construct the maps $\pi$. Let us denote by $Q^{(\mathrm{e},B)}_{p,D}$ and $Q^{(\mathrm{e},\tilde{B})}_{q,D}$ the electric-like asymptotic charges associated with forms $B \in \Omega_{\mathrm{AS}}^{p}(M_D)$ and $\tilde{B} \in \Omega_{\mathrm{AS}}^{q}(M_D)$, respectively. The idea is that $\pi_1$ associates to each form $B \in \Omega_{\mathrm{AS}}^{p}(M_D)$ with charge $Q^{(\mathrm{e},B)}_{p,D}$, given by \eqref{chargeprad}, a vector $v \in \mathbb{C}^{n_p}$ defined by
\begin{equation}
    v:=\begin{cases}
        \pi_1(B)=0 \ \ \ \textnormal{if} \  B\equiv 0;\\
        \pi_1(B)=Q^{(\mathrm{e},B)}_{p,D}e_1+\sum_{k=2}^{n_p} b^{(B)}_ke_k \ \ \ \textnormal{otherwise},
    \end{cases}
\end{equation}
where $e_1,...,e_{n_p}$ is an orthonormal basis of $\mathbb{C}^{n_p}$ and $b^{(B)}_k\in \mathbb{R}$ are some of the independent entries of the coordinate representation of the form, chosen in such a way that different forms correspond to different string objects. 
Moreover, for all $B,C\in \Omega_{\mathrm{AS}}^{p}(M_D)$ we have
\begin{equation}
\begin{aligned}
    &\pi_1(B+C)=Q^{(\mathrm{e},B+C)}_{p,D}e_1+\sum_{k=2}^{n_p}(b^{(B)}_k+b^{(C)}_k)e_k,\\
    &\pi_1(\lambda B)=Q^{(\mathrm{e},\lambda B)}_{p,D}e_1+\sum_{k=2}^{n_p}(\lambda b^{(B)}_k)e_k, \ \ \lambda \in \mathbb{C},
\end{aligned}
\end{equation}
hence the map is linear. Furthermore, it is injective and surjective by construction, therefore it is bijective and thus invertible. The same considerations hold for the map $\pi_2$, which is defined similarly to $\pi_1$, and therefore also for $\pi_2|_{\mathbb{C}^{n_p}}$, where $\mathbb{C}^{n_p} \subset \mathbb{C}^{n_{q}}$ is a subspace whose first component is given by $Q^{(\mathrm{e},\Tilde{B})}_{q,D}$ for some $\tilde{B} \in \Omega_{\mathrm{AS}}^{q}(M_D)$. Since $\pi_2|_{\mathbb{C}^{n_p}}$ is a bijection, the subspace $\mathbb{C}^{n_p} \subset \mathbb{C}^{n_{q}}$ is mapped by $\pi_2^{-1}|_{\mathbb{C}^{n_p}}$ into a subspace of $\Omega_{\mathrm{AS}}^{q}(M_D)$ of dimension $n_p$; the operator $\star_{D-2}$ maps $\Omega_{\mathrm{AS}}^{p}(M_D)$ into this subspace and can be defined as
\begin{equation}
  \star_{D-2} := d^{-1} \circ \star \circ d.
\end{equation} 
We know that $d$ is a linear bijection because it is injective (since we have removed forms with vanishing exterior derivative, i.e. pure gauge configurations for which the field strength vanishes) and surjective (since we have removed field configurations that differ by pure gauge, thanks to the quotient in the definition of asymptotic symmetries). Moreover, $\star$ is an isomorphism by definition, and therefore $\star_{D-2}:=d^{-1} \circ \star \circ d$ is a linear bijection from a subspace of $\Omega_{\mathrm{AS}}^{p}(M_D)$ to a subspace of $\Omega_{\mathrm{AS}}^{q}(M_D)$. 

In this setup, the existence and uniqueness of the duality map $f$ follows from the good behaviour of $\star_{D-2}$, $\pi_1$ and $\pi_2$. Indeed, since the $\pi$ maps are bijective, they are invertible and the duality map can be defined as $f:=\pi_2 \circ \star_{D-2} \circ \pi_1^{-1}$. The map $f$ itself is a bijection, since it is a composition of bijections. The duality map $f$ is therefore realized as an automorphism of $\mathbb{C}^{n_p}$, so $f \in \mathrm{Aut}(\mathbb{C}^{n_p})=\mathrm{GL}(\mathbb{C}^{n_p}) \cong \mathrm{GL}(n_p,\mathbb{C})$. 

Let us write a generic transformation $f$ in $\mathrm{GL}(n_p,\mathbb{C})$ as
\begin{equation}
f=\begin{bmatrix}
\eta & \Vec{\alpha} \\
\Vec{\beta}^{\mathrm{t}} & A  \\
\end{bmatrix};
\end{equation}
with $\eta \in \mathbb{C} \setminus \{0\}$, $A \in \mathrm{GL}(n_p-1,\mathbb{C})$ and $\Vec{\alpha},\Vec{\beta} \in \mathbb{C}^{n_p-1}$. Applying this transformation to the vector $v=\pi_1(B)=Q^{(\mathrm{e},B)}_{p,D}e_1+\sum_{k=2}^nb^{(B)}_ke_k=(Q^{(\mathrm{e},B)}_{p,D},\Vec{b}^{(B)})$, we obtain a vector 
\begin{equation}
\Tilde{v}=\pi_2(\Tilde{B})=Q^{(\mathrm{e},\Tilde{B})}_{q,D}e_1+\sum_{k=2}^n\Tilde{b}^{(\Tilde{B})}_ke_k=(\Tilde{Q}^{(\mathrm{e},\Tilde{B})}_{q,D}, \Vec{\tilde{b}}^{(\Tilde{B})}),
\end{equation}
with $\Tilde{b}^{(\Tilde{B})}_k \in \mathbb{R}$ and $\Tilde{B}=\star_{D-2}B \in \Omega_{\mathrm{AS}}^{q}(M_D)$. 

Therefore
\begin{equation}
\begin{aligned}
    &\eta Q^{(\mathrm{e},B)}_{p,D}+\Vec{b}^{(B)}\cdot \Vec{\alpha}=Q^{(\mathrm{e},\Tilde{B})}_{q,D},\\
    &\Vec{\beta}Q^{(\mathrm{e},B)}_{p,D}+A \cdot \Vec{b}^{(B)}=\Vec{\tilde{b}}^{(\Tilde{B})}.
\label{eqmapping}    
\end{aligned}
\end{equation}
Since an invertible matrix is triangularizable if and only if its characteristic polynomial has all roots in the underlying field, and $\mathbb{C}$ is algebraically closed, we can always choose, by a change of basis, $f$ such that $\Vec{\alpha}=0$ in \eqref{eqmapping}. 
The complex parameter $\eta$, which depends on the dimension $D$ and on the degrees of the forms involved, is simply given by 
\begin{equation}
\eta=\frac{Q^{(\mathrm{e},\tilde{B})}_{q,D}}{Q^{(\mathrm{e},B)}_{p,D}};
    \label{eta}
\end{equation}
equivalently,
\begin{equation}
Q^{(\mathrm{e},\tilde{B})}_{q,D}=f|_Q(Q^{(\mathrm{e},B)}_{p,D})
\end{equation}
where $f|_Q(\bullet)=\eta \bullet$.\\
Explicitly, using \eqref{chargeprad} and \eqref{eta}, we have
\begin{equation}
\begin{aligned}
    &\Lambda_q\oint_{S_u^{D-2}}\gamma^{i_1j_1}...\gamma^{i_{q-1}j_{q-1}}    \tilde{\epsilon}_{i_1...i_{q-1}}^{(\frac{D-(2q+2)}{2})} \tilde{\mathcal{R}}^{(q)}d\Omega=\Lambda_p\oint_{S_u^{D-2}}\gamma^{i_1j_1}...\gamma^{i_{p-1}j_{p-1}} \epsilon_{i_1...i_{p-1}}^{(\frac{D-(2p+2)}{2})}\mathcal{R}^{(p)}d\Omega;
\end{aligned}    
\end{equation}
bringing everything to one side we obtain a vanishing charge, but it cannot be associated with an identically vanishing form, since for $p\neq q$ we cannot isolate the difference $ \tilde{\mathcal{R}}^{(q)}-{\mathcal{R}}^{(p)}$. Therefore we must have, explicitly,
\begin{equation}
\begin{aligned}
    &\Lambda_{q} \gamma^{j_1\Tilde{j}_1}...\gamma^{j_{q-1}\Tilde{j}_{q-1}} \tilde{\epsilon}_{j_1...j_{q-1}}^{(\frac{D-(2q+2)}{2})} \tilde{\mathcal{R}}^{(q)}=\Lambda_{p}\gamma^{i_1j_1}...\gamma^{i_{p-1}j_{p-1}} \epsilon_{i_1...i_{p-1}}^{(\frac{D-(2p+2)}{2})} {\mathcal{R}}^{(p)}.
    \label{mappa1}
\end{aligned}
\end{equation}

If $p= q$, i.e. $p=\frac{D-2}{2}$, the charges involve forms of the same degree and we can isolate their difference, which must be the identically vanishing form. We obtain 
\begin{equation}
    \tilde{\mathcal{R}}^{(\frac{D-2}{2})}- {\mathcal{R}}^{(\frac{D-2}{2})}=0;
    \label{mappa2}
\end{equation}
which means that the charge is mapped to a multiple of itself, and hence the form is self-dual. 
\end{proof}

Therefore, the duality map is topological in nature and can be constructed if and only if 
\begin{equation}
    H^p=H^{p+1}=0=H^{q+1}=H^{q}.
    \label{cohom}
\end{equation}
Indeed, the vanishing of these cohomology groups is sufficient to reduce the full de Rham complex to the asymptotic symmetries de Rham complex, and it is also necessary, since to construct the duality map we need $d_p$ and $d_q$ to be isomorphisms. Again, we emphasize that $d_p^{-1}$ and $d_q^{-1}$ are defined only on the image, respectively, of $d_p$ and $d_q$, and that their definition relies on the vanishing of the de Rham cohomology groups \eqref{cohom}. We observe that similar steps can be retraced in the case of Coulomb fall-off and $D=D_c=2p+2$, and hence also in this case the theorem above holds.

Some comments are in order. First of all, let us discuss what happens if the assumptions of the theorem are weakened. The requirement that the fall-offs produce well-defined charges is essential in order to exclude pathological behaviours in the charges that would make it impossible to relate the charges in the two dual descriptions via a linear map (at least in the sense of complex vector spaces).\\ The hypothesis of trivial topology (or at least the requirement \eqref{cohom}) is the most delicate one, since it ensures exactness and thus the existence of $d_p^{-1}$ and $d_q^{-1}$. When the spacetime has nontrivial topology, such as in the presence of wormholes, handles or defects, the operators $d_p$ and $d_q$ are no longer invertible and $d_p^{-1}$ and $d_q^{-1}$ cease to exist. One reason is, for example, that if $H^{p} \neq 0$, then, by the Hodge theorem, non-vanishing harmonic forms appear in $\Omega^{p}$. Since a harmonic form is closed and corresponds to a unique cohomology class in $H^{p}$, it cannot arise from a $(p-1)$-form unless the harmonic form vanishes. This means that it is not possible to reduce the de Rham complex to its asymptotic symmetries version; hence $d_p$ cannot be an isomorphism and it cannot be inverted. Moreover, we underline that the argument presented above requires only the existence of $d_p^{-1}$ and $d_q^{-1}$; it is therefore not necessary to choose any specific homotopy operator.\\
Finally, let us briefly comment on how the de Rham complexes on $M_D$ relate to complexes on constant-$u$ cuts $S_u^{D-2}$. The geometrical setting is as follows. The constant-$u$ cuts $S_u^{D-2}$ define submanifolds of the conformally compactified spacetime; therefore we have a canonical inclusion map
\begin{equation}
i_u : S_u^{D-2} \hookrightarrow M_D .
\end{equation}
This allows us to restrict bulk differential forms to the boundary via the pullbacks
\begin{equation}
i_u^*:\Omega^p(M_D)\to\Omega^p(S_u^{D-2});
\end{equation}
geometrically, $i_u^*\omega$ is obtained by keeping only the components of $\omega$ tangent to the sphere and evaluating them in the limit $r\to+\infty$ at fixed $u$. However, there is a caveat: the pullback kills all components with $r$ or $u$ indices, since $i_u^*(dr)=i_u^*(du)=0$. Therefore, in such cases, the correct procedure is to first contract the form with the vector fields $\partial_r$ and/or $\partial_u$ and then apply the pullback via the inclusion map:
\begin{equation}
\begin{aligned}
&i_u^*\circ \iota_{\partial_r}:\Omega^p(M_D)\to\Omega^{p-1}(S_u^{D-2});\\
&i_u^*\circ \iota_{\partial_u}:\Omega^p(M_D)\to\Omega^{p-1}(S_u^{D-2});\\
&i_u^*\circ  \iota_{\partial_u} \circ \iota_{\partial_r}:\Omega^p(M_D)\to\Omega^{p-2}(S_u^{D-2});\\
\end{aligned}
\end{equation}
where $\iota_{X}$ denotes contraction with the vector field $X$. For example, the components $H_{rui_{1}...i_{p-1}}$ are the components of a $(p-1)$-form on $S_u^{D-2}$. Hence there is no obvious relation between the de Rham complex of $M_D$ and that of $S_u^{D-2}$. However, it would be possible to construct an ad hoc de Rham complex that accommodates this structure, keeping track of the normal directions to the sphere, and which could allow us to understand the map $\star_{D-2}$ directly and intrinsically in terms of boundary data, that is, on the differential forms of this extended de Rham complex over $S_u^{D-2}$. The explicit construction of this complex and the reinterpretation of the map $\star_{D-2}$ solely form boundary data can be material for another work, with the scope to better understand the nature of the duality.

\subsubsection{Physical interpretation}\label{phyint}
Let us no discuss the result of theorem \eqref{THM3.1} from a physical point of view. First of all, when the spacetime satisfies the conditions \eqref{cohom}, and hence we can write $Q_{q,D}^{(\mathrm{e})}=\eta Q_{p,D}^{(\mathrm{e})}$, magnetic-like and electric-like charges act equivalently on physical states. Indeed, to promote charges to operators acting on the Hilbert space $\mathcal{H}$ we use continuous functional calculus and, in the end, the electric-like and magnetic-like charge operators act differently on a vector in the Hilbert space only by an overall multiplicative factor. However, they act in the same way on the projective Hilbert space, i.e. the state space or ray space, since it is defined by quotienting by the equivalence relation $|\psi\rangle \sim_{\text{eq}} \lambda |\psi\rangle$ with $|\psi\rangle \in \mathcal{H}, \ \lambda \in \mathbb{C}$. This is consistent with the fact that the two theories describe the same degrees of freedom. 

Indeed, we can construct electric-like and magnetic-like Hilbert vectors using, respectively, the $p$-form gauge field operator $B^{(p)} \in \Omega^p$ and its dual $\tilde{B}^{(q)} \in \Omega^q$:
\begin{equation}
    |\mathrm{e}\rangle_{B^{(p)}}, \ \ \ \ |\mathrm{m}\rangle_{\tilde{B}^{(q)}};
\end{equation}
but we can also construct another pair of electric-like and magnetic-like Hilbert vectors using, respectively, the dual $q$-form gauge field operator $B^{(q)} \in \Omega^q$ and its dual $\tilde{B}^{(p)} \in \Omega^p$: 
\begin{equation}
    |\mathrm{e}\rangle_{B^{(q)}}, \ \ \ \ |\mathrm{m}\rangle_{\tilde{B}^{(p)}}.
\end{equation}
However, by definition, the dual field $\tilde{B}^{(q)}$ coincides with the field $B^{(q)}$, hence
\begin{equation}
|\mathrm{e}\rangle_{B^{(q)}}=|\mathrm{m}\rangle_{\tilde{B}^{(q)}}
\end{equation}
and, by the same reasoning (up to the scalar factors needed to invert the Hodge star operator), we have 
\begin{equation}
|\mathrm{e}\rangle_{B^{(p)}}=(-1)^{(p+1)(q+1)}s|\mathrm{m}\rangle_{\tilde{B}^{(p)}}.
\end{equation}
Moreover, since the duality\footnote{Recall that our notation is such that $dB^{(q)}=\star dB^{(p)}$ and $dB^{(p)}=(-1)^{(p+1)(q+1)}s\star dB^{(q)}$.} relates $B^{(p)}$ with $\tilde{B}^{(q)}$ and $B^{(q)}$ with $\tilde{B}^{(p)}$, we obtain 
\begin{equation}
|\mathrm{e}\rangle_{B^{(p)}}=(-1)^{(p+1)(q+1)}s \alpha |\mathrm{m}\rangle_{\tilde{B}^{(q)}}=(-1)^{(p+1)(q+1)}s \alpha|\mathrm{e}\rangle_{B^{(q)}}
\end{equation}
and 
\begin{equation}
|\mathrm{e}\rangle_{B^{(q)}}=(-1)^{(p+1)(q+1)}s \beta |\mathrm{m}\rangle_{\tilde{B}^{(p)}}=\beta|\mathrm{e}\rangle_{B^{(p)}}
\end{equation}
where $\alpha \in \mathbb{C}$ and, by consistency,
\begin{equation}
    \beta^{-1}=(-1)^{(p+1)(q+1)}s \alpha.
\end{equation}
Therefore, putting all the pieces together, we find
\begin{equation}
Q_{p,D}^{(\mathrm{e})}|\mathrm{e}\rangle_{B^{(p)}}=Q_{p,D}^{(\mathrm{e})}(-1)^{(p+1)(q+1)}s \alpha|\mathrm{e}\rangle_{B^{(q)}},
\end{equation}
but the appropriate charge that should act on the vector $|\mathrm{e}\rangle_{B^{(q)}}$ is $Q_{q,D}^{(\mathrm{e})}$, hence
\begin{equation}
    Q_{q,D}^{(\mathrm{e})}=Q_{p,D}^{(\mathrm{e})}(-1)^{(p+1)(q+1)}s \alpha \Rightarrow Q_{q,D}^{(\mathrm{e})}=\eta Q_{p,D}^{(\mathrm{e})}
\end{equation}
with $\eta:=(-1)^{(p+1)(q+1)}s \alpha \in \mathbb{C}.$ However, we know that the electric-like and magnetic-like charges of the two dual descriptions are related by \eqref{eqelecmagn} and, in order to be consistent with that relation, we must fix $\eta=1.$

\subsection{Comments on the duality for power law divergent/vanishing charges}\label{dualmap2}
In this brief section we present some results for asymptotic charges with Coulomb fall-off \eqref{chargepcoul}. In this section we treat logarithmic terms as pure gauge, and therefore we do not include them in the charges.

First of all, let us discuss the magnetic-like charge in this case. For the charges \eqref{chargepcoul} involving field components with Coulomb fall-offs one can
formally reproduce the same steps of the radiation fall-offs case, however, we have to carefully take into account the different orders entering the charges. The result is a magnetic-like which is divergent in $D > 2p+ 2 $ and vanishing in $D < 2p+ 2$, which is
the opposite behavior with respect to the charges \eqref{chargepcoul}: when the Coulombic
electric-like charge \eqref{chargepcoul} vanishes the corresponding magnetic-like charge diverges, and viceversa. Therefore a charge that in a given theory is subleading with respect to the charge with radiation fall-offs is mapped to an overleading charge in the dual
description. One simple instance is the dual pair two-form and scalar field in
$D=4$. In that case the overleading $\mathcal{O}(r^2)$ 2-form charge is mapped to a subleading scalar charge of order $\mathcal{O}(r^{-2})$. Scalar charges acting at $\mathcal{O}(r^{-2})$ appear, for example, in \cite{Hamada_2017} in connection with the pion memory effect.

As for the duality map, we present some speculative results, specifically Conjecture \ref{Conj3.1}. The main difference with respect to the case of Theorem \ref{THM3.1} is that we need to add the point at infinity to $\mathbb{C}^n$; we denote by $\overline{\mathbb{C}^n}=\mathbb{C}^n \cup \{\infty\}$ the projective (or Aleksandrov) compactification of $\mathbb{C}^n$, also known as the Riemann sphere. Let us focus on the case $D<2p+2$, since for $D>2p+2$ we have the opposite picture. The duality map relates the asymptotic charge of a $p$-form with that of a $(q=D-p-2)$-form; this means that when the asymptotic charge of the $p$-form is power-law divergent, that of the dual $q$-form is power-law vanishing, and viceversa\footnote{Indeed the asymptotic charge of a $(q=D-p-2)$-form scales as $\mathcal{O}(r^{D-2p-2})$, while that of a $p$-form scales as $\mathcal{O}(r^{-(D-2p-2)})$.}. Hence the duality maps a vanishing charge into a divergent one. \\
To have a more precise picture, we first associate the infinite charge with the point at infinity and, second, the vanishing charge with the origin. The duality map can then be identified with the reciprocal map, $R(v_1,...,v_n)=(\frac{1}{v_1},...,\frac{1}{v_n})$, which is a total function on the projectively extended space\footnote{This structure allows for division by zero and the presence of the point at infinity.}. Therefore every component of a vector in one copy of $\overline{\mathbb{C}^n}$ is mapped to its inverse in the other copy of $\overline{\mathbb{C}^n}$. If we partially follow the idea in the proof of Theorem \ref{THM3.1}, we can restrict the duality map to its first component in order to obtain the duality map between the charges:
\begin{equation}
    Q^{(\mathrm{e},\Tilde{B})}_{q,D}=\frac{1}{Q^{(\mathrm{e},B)}_{p,D}}.
    \label{mapcol}
\end{equation}
This duality map is a special case of a Möbius transformation\footnote{This is consistent with the fact that now, by a modification of Theorem \ref{THM3.1}, we have $f \in \mathrm{Aut}(\overline{\mathbb{C}^n})$, and Möbius transformations are precisely the automorphisms of the Riemann sphere.}, but the particularly interesting point is that a power-law vanishing charge, i.e. a gauge theory with trivial gauge transformations at null infinity, is mapped to a power-law divergent charge, i.e. a gauge theory with boundary conditions that are too weak in the sense of the following definition.

\begin{Definition}[Power law weak fall-offs]
Let $(J, M, G)$ be an abelian gauge theory with $(M,\boldsymbol{g})$ a $D$-dimensional Lorentzian manifold with conformal boundary $\partial M_c$, $J$ a set of fields and $G$ an abelian group. The set of boundary conditions $J|_{\partial M_c}$ is said to be power law weak if the asymptotic charge computed on $\partial M_c$ with field fall-off given by $J|_{\partial M_c}$ is power-law divergent. 
\end{Definition}

Therefore, a $p$-form gauge theory with Coulomb fall-offs has trivial gauge transformations at null infinity if and only if the dual $q$-form gauge theory has power law weak boundary conditions. We can attempt to generalize this observation to generic abelian gauge theories and their duals with the following conjecture.

\begin{Conjecture}[Link between trivial gauge transformations and power law weak fall-offs]\label{Conj3.1} 
    Let $(M,\boldsymbol{g})$ be a $D$-dimensional Lorentzian manifold with conformal boundary $\partial M_c$ and let $(J, M, G)$ be an abelian gauge theory with gauge group $G$ and fields boundary conditions $J|_{\partial M_c}$. The dual gauge theory $(\Tilde{J}, M, G)$ has trivial gauge transformations at the conformal boundary $\partial M_c$ if and only if  $J|_{\partial M_c}$ is power law weak.
\end{Conjecture}

Let us return to the issue of renormalizing power-law divergent charges. The possibility of renormalizing such charges in one description means, in view of the above conjecture, that we have to eliminate the trivial gauge transformations in the dual description. This is meaningful, since renormalizing a divergent charge amounts to removing divergences that are not physical, and this can be achieved by removing the unphysical (i.e. redundant) gauge transformations from the dual description. 

\section{Asymptotic charges and $p$-form symmetries}\label{sec4}

In this section we propose a novel framework connecting $p$-form gauge theories on Monkowski spacetime with their asymptotic charges and dualities. The central idea is to interpret asymptotic charges $Q^{(\mathrm{e})}_{p,D}$ and their magnetic-like duals as integrated operators on the celestial sphere, generating higher-form symmetries in the celestial CFT. This perspective may lead to the existence of new celestial current algebras and cohomological structures.

Let us consider a $p$-form gauge field $B$, its field strength $H=dB$, and the dual gauge theory with a $q$-form field $\tilde{B}$ and field strength $\star H=:\tilde{H}=d\tilde{B}$, where $q=D-p-2$. The asymptotic gauge parameters of the two formulations are, respectively, $\epsilon$ and $\tilde{\epsilon}$, with the conditions $\partial_u \epsilon=\partial_u \tilde{\epsilon}=0$ at leading order in an $r$-power-law expansion.\\
We can construct the following object:
\begin{equation}
   Q^{(p)}(\Sigma_u^{p+1}):= \int_{\Sigma_u^{p+1}} \star H
   \label{hf}
\end{equation}
where $\Sigma_u^{p+1} \subset S_u^{D-2}$. This charge is conserved since the local $(q+1)$-current $j^{(q+1)}:= \star H$ is closed as a consequence of the relation $d\star H=0$. This object can be regarded as the conserved charge associated with a codimension $D-p-1=q+1$ submanifold; hence it is a conserved charge associated with a $p$-form symmetry, following the terminology of \cite{Gaiotto:2014kfa}. 

At this point we can take the local current $j^{(q+1)}$ and construct a smeared charge on the null-infinity sphere by pairing it with an appropriate form. Since the total form degree must be $D-2$, we need a $D-2-(q+1)=D-q-3=p-1$ form. The natural choice is the gauge parameter ${\epsilon} \in \Omega^{p-1}(S_u^{D-2})$ of the $p$-form gauge field:
\begin{equation}
    Q_{\mathrm{smeared}}=\int_{S^{D-2}_u} {\epsilon} \wedge j^{(q+1)};
    \label{sme}
\end{equation}
since the asymptotic charge \eqref{caricapforma} is also given by an integral of a $(D-2)$-form, these two objects must be proportional. Therefore, the 
$r \rightarrow +\infty$ limit of the smeared charge $Q_{\mathrm{smeared}}$ must be proportional to the asymptotic charge computed in Section \ref{anasin}, with the radiative field component entering the construction.

The smeared charge \eqref{sme} reproduces the standard higher-form symmetry charge \eqref{hf} under specific conditions on the support of $\epsilon$ and on the choice of submanifold. Fix a retarded time $u$ and let
$\iota: \Sigma^{p+1}_u \hookrightarrow S^{D-2}_u$ be a smooth, oriented, closed submanifold of
codimension $q+1 = D - p - 1$. We choose $\epsilon \in \Omega^{p-1}(S^{D-2}_u)$
to be a smooth representative of the Poincar\'e dual of $\Sigma^{p+1}_u$, i.e. a
$(p-1)$-form supported in a small tubular neighbourhood of $\Sigma^{p+1}_u$, normalized so that
\begin{equation}
  \int_{S^{D-2}_u} \epsilon \wedge \alpha
  = \int_{\Sigma^{p+1}_u} \iota^* \alpha
\end{equation}
for any smooth $(q+1)$-form $\alpha$ on $S^{D-2}_u$. Taking
$\alpha = j^{(q+1)} = \star H$, we obtain
\begin{equation}
  Q_{\text{smeared}}[\epsilon]
  = \int_{S^{D-2}_u} \epsilon \wedge j^{(q+1)}
  = \int_{\Sigma^{p+1}_u} \star H
  = Q^{(p)}(\Sigma^{p+1}_u) .
\end{equation}
Thus, when $\epsilon$ is chosen
as the Poincar\'e dual of $\Sigma^{p+1}_u$, the smeared charge reproduces the
standard $p$-form symmetry charge. Since $j^{(q+1)}$ is closed, $dj^{(q+1)}=0$,
the charge depends only on the homology class of $\Sigma^{p+1}_u$, as expected
for a topological operator.

In the standard language of generalized global symmetries, $j^{(q+1)}$ plays the role
of the conserved current associated with a $p$-form symmetry, and the operator
$U_\alpha(\Sigma^{p+1}_u) = \exp\bigl(i \alpha Q^{(p)}(\Sigma^{p+1}_u)\bigr)$
defines a codimension $(p+1)$ topological surface operator. For the $p$-form
gauge field $B$, the charged objects are Wilson surfaces
$W_e(C_p) = \exp(i \int_{C_p} B)$ supported on $p$-dimensional submanifolds
$C_p$. When $\Sigma^{p+1}_u$ links $C_p$, the operator $U_\alpha(\Sigma^{p+1}_u)$
acts on $W_e(C_p)$ by a phase determined by the intersection number, in complete
agreement with the standard picture of higher-form symmetries and their
topological operators.

Moreover, $Q_{\mathrm{smeared}}$ can also be regarded as an integrated, $u$-dependent form current on the celestial sphere, which takes an asymptotic gauge parameter as input and returns a scalar obtained by pairing with the field strength of the dual formulation:
\begin{equation}
\mathcal{J}^{(q+1)}[\bullet](u) \;:=\; \int_{S^{D-2}_u} \bullet \wedge j^{(q+1)}\, .
\end{equation}
We can further integrate over retarded time $u$ in order to capture the soft contribution (i.e. the soft modes as $\omega \to 0$):
\begin{equation}
\mathcal{J}^{(q+1)}[\bullet] \;:=\; \int_{-\infty}^{+\infty} \! du \;\, \mathcal{J}^{(q+1)}[\bullet](u)  \, .
\end{equation}
The form functional $\mathcal{J}^{(q+1)}[\bullet]: \Omega^{p-1}(S^{D-2}) \rightarrow \mathbb{R}$ is then interpreted as an integrated $(q+1)$-form current on the celestial sphere $S^{D-2}$ associated with the soft modes of a $p$-form gauge field. When $\mathcal{J}^{(q+1)}$ is evaluated on the gauge parameter $\epsilon$, it returns a scalar that coincides with the $u$-integrated smeared charge. 

Analogous considerations apply to the charge
\begin{equation}
   Q^{(q)}(\Sigma_u^{q+1}):= \int_{\Sigma_u^{q+1}} H,
   \label{hf2}
\end{equation}
where the pairing must now be performed with $\tilde{\epsilon} \in \Omega^{q-1}(S_u^{D-2})$. We denote by $\tilde{Q}_{\mathrm{smeared}}$ the associated smeared charge and by $\mathcal{J}^{(p+1)}[\bullet]$ the corresponding form functional, interpreted as an integrated $(p + 1)$-form current on the celestial sphere $S^2$ associated with the soft modes of a $q$-form gauge field. 

This construction could provide natural candidates for celestial currents associated with electromagnetic-like charges. However, a more careful analysis, based on an explicit oscillator expansion of the fields, extraction of the soft modes, and a Mellin transform, would be required to identify the precise form of the celestial currents. The working hypothesis is that these currents are proportional to a formal complex combination of $\mathcal{J}^{(p+1)}[\bullet]$ and $\mathcal{J}^{(q+1)}[\bullet]$. A detailed investigation of this point is left for future work and it will probably require results in \cite{Donnay:2022ijr}.

These constructions partially address open question $\#3$ of the “open questions of the Simons collaboration on celestial holography” \cite{simons}. Let us clarify this statement. The program of systematically relating asymptotic symmetries to generalized (higher-form) symmetries has its roots in earlier work. In abelian $p$-form theories, the author of \cite{Lake:2018dqm} clarified that the conserved quantities associated with higher-form symmetries have essentially the same structure as the large gauge charges familiar from asymptotic symmetry analyses. Subsequent studies of asymptotic symmetries for dynamical $p$-form gauge fields \cite{Afshar:2018apx,Francia:2018jtb} showed that, at least in certain settings, large $p$-form gauge transformations do carry well-defined charges. However, in those analyses the connection to higher-form symmetries was only implicit: although the charges were derived, they were neither explicitly identified with higher-form symmetry charges nor framed as local refinements of the standard global higher-form charges. Independent developments in holography \cite{Hofman_2018} further suggest that global higher-form symmetries of the boundary theory correspond to bulk gauge invariances of form gauge fields, with boundary charges represented as surface integrals of the bulk field strength. This again points to a common origin with asymptotic charges. 

The present work shows that the smeared charge \eqref{sme} is, on the one hand, proportional to the asymptotic charge and, on the other hand, reproduces the higher-form symmetry charge. Indeed, the $p$-form symmetry charge associated with the local current $j^{(q+1)}=\tilde{H}$ is naturally correlated with a smeared charge constructed by pairing with the $p$-form gauge parameter $\epsilon$, and this charge must be proportional to the asymptotic charge computed earlier. Note also that the smeared charge reproduces the $p$-form symmetry charge if $\epsilon$ is chosen to be constant and supported on $\Sigma^{(p+1)}$. Essentially, up to numerical factors, the asymptotic charge of a $p$-form gauge theory reproduces the $p$-form symmetry charge when the gauge parameter is chosen to be constant and non-vanishing only on a codimension $q+1$ submanifold.\\
These observations suggest that the asymptotic charge algebra may be interpreted as a boundary-refined version of the generalized symmetry algebra, with both emerging from the same underlying geometric structure. In this sense, the results presented here strengthen earlier observations and provide a possible explicit bridge between asymptotic and generalized symmetries.

We also note that the $u$-variation of the smeared charge \eqref{sme} is proportional to the gauge variation of the $p$-form gauge field. Indeed, consider the cylindrical region 
\(\mathcal C_{[u_1,u_2]}\) bounded by the two spheres 
\(S^{D-2}_{u_1}\) and \(S^{D-2}_{u_2}\). By Stokes’ theorem,
\begin{equation}
    \int_{S^{D-2}_{u_2}} \epsilon \wedge j^{(q+1)}
      - \int_{S^{D-2}_{u_1}} \epsilon \wedge j^{(q+1)}
    = \int_{\mathcal C_{[u_1,u_2]}} d\!\left( \epsilon \wedge j^{(q+1)} \right) \, ,
\end{equation}
Using the identity 
\(d(\alpha \wedge \beta)= d\alpha \wedge \beta + (-1)^{\deg \alpha}\alpha \wedge d\beta\)
and the fact that \(dj^{(q+1)}=0\), we obtain
\begin{equation}
    Q_{\mathrm{smeared}}(u_2)-Q_{\mathrm{smeared}}(u_1)
    = \int_{\mathcal C_{[u_1,u_2]}} d\epsilon \wedge j^{(q+1)} \, .
\end{equation}
Therefore, the charge is conserved if one of the following conditions holds:
\begin{itemize}
    \item \(d\epsilon=0\) on the cylinder;
    \item the support of \(d\epsilon\) does not intersect the support of \(j^{(q+1)}\);
    \item the form \(d\epsilon \wedge j^{(q+1)}\) is exact on the cylinder and the boundary contribution vanishes due to boundary conditions.
\end{itemize}
In our setup, however, \(d\epsilon\) coincides with the gauge variation \(\delta_\epsilon B\) of the $p$-form gauge field. Hence,
\begin{equation}
    Q_{\mathrm{smeared}}(u_2)-Q_{\mathrm{smeared}}(u_1)
    =\int_{\mathcal C_{[u_1,u_2]}} \delta_\epsilon B \wedge j^{(q+1)} \, .
\end{equation}
The failure of conservation can be interpreted as an exchange of flux between the $p$-form field and its dual. What appears as a non-conserved smeared charge in one description may correspond to a flux of the dual field strength in the other. Thus conservation is not lost, but rather redistributed between dual sectors. 

Since $Q_{\mathrm{smeared}}$ and $\tilde{Q}_{\mathrm{smeared}}$ are proportional to the asymptotic charges $Q_{p,D}^{(\mathrm{e})}$ and $Q_{q,D}^{(\mathrm{e})}$, it follows from the theorem proved above that there exists a unique map (under the appropriate topological assumptions) such that $Q_{\mathrm{smeared}}$ is proportional to $\tilde{Q}_{\mathrm{smeared}}$ for every choice of $\epsilon$ and $\tilde{\epsilon}$. If we choose the gauge parameters to have support only on $\Sigma^{(p+1)}$ and $\Sigma^{(q+1)}$, respectively, we obtain the relation
\begin{equation}
Q^{(p)}(\Sigma_u^{p+1})\propto Q^{(q)}(\Sigma_u^{q+1}) \quad\Rightarrow \quad \int_{\Sigma_u^{p+1}} \star H \propto \int_{\Sigma_u^{q+1}} H;
\end{equation}
hence the two cycles $\Sigma_u^{q+1}$ and $\Sigma_u^{p+1}$ are related by Poincaré duality in homology.

We define the charge algebra via its action on fields, namely
\begin{equation}
  \{Q_{\mathrm{smeared}}[\epsilon],Q_{\mathrm{smeared}}[\tilde{\epsilon}]\}
  =\delta_{\epsilon}Q_{\mathrm{smeared}}[\tilde{\epsilon}]
  \;=\;\int_{S^{D-2}_u}\tilde{\epsilon}\wedge\delta_{\epsilon} j^{(q+1)} 
  \;+\;\int_{S^{D-2}_u}(\delta_{\epsilon}\tilde{\epsilon})\wedge j^{(q+1)}.
\end{equation}
Since $\delta_{\epsilon}\tilde{\epsilon}=0$ and $j^{(q+1)}$ is gauge invariant under the $p$-form gauge symmetry, because the Hodge operator $\star$ depends only on the metric and not on the gauge degrees of freedom, we have
\begin{equation}
  \delta_{\epsilon} j^{(q+1)} =0 \;\;\Rightarrow\;\; 
  \{Q_{\mathrm{smeared}}[\epsilon]\,,\,Q_{\mathrm{smeared}}[\tilde{\epsilon}]\}=0 .
\end{equation}
Therefore, the smeared charges generate an abelian algebra. Possible non-abelian structures may nevertheless emerge in specific situations. One such situation arises in the presence of gravitational dressing \cite{Giddings_2019,Donnelly_2016+,Giddings_2018,Giddings:2025bkp}. In a quantum theory of gravity, even perturbatively, the definition of a local operator such as $\mathcal{O}(x)$ is meaningless, since the theory must be invariant under diffeomorphisms and the point $x$ has no invariant physical meaning. The resolution is to construct a gravitationally dressed operator. Starting from a bare local field \( \mathcal{O}(x) \), one attaches a geometric displacement \( V^\mu(x) \) generated by the linearized gravitational field, and defines the physical operator as \( \mathcal{O}_{\text{dressed}}(x) = \mathcal{O}\left(x + V(x)\right) \). A particularly simple example is the gravitational line dressing \cite{Donnelly_2016++}, where \( V^\mu(x) \) depends on an arbitrary curve connecting $x$ to infinity and on the graviton field $h_{\mu \nu}$, which encodes perturbations around flat space. The dressed operator is then expanded perturbatively in the graviton field (i.e. in $V(x)$): the zeroth-order term is the bare operator, while higher-order terms involve derivatives of the bare operator and non-local graviton contributions. Gravitational dressing deforms the asymptotic charge algebra, since the theory effectively becomes a coupled $p$-form–Einstein system: the gauge charges remain abelian but no longer decouple from BMS charges, and the full charge algebra may therefore become non-abelian. Moreover, gravitational dressing can be interpreted as a line or surface operator that modifies the gravitational boundary conditions along a submanifold, acting as a topological defect. Such defects implement generalized symmetries, since operators that link them transform non-trivially under shifts of soft gravitational modes, and the corresponding fusion rules can become non-invertible. 

Another possible source of non-abelian structure could arise in the presence of mixed-symmetry tensor gauge fields, which naturally appear in the tensionless limit of the string spectrum. Two key features are relevant here. First, a mixed-symmetry tensor gauge field can admit more than one gauge parameter, due to its Young-tableau symmetry structure. Second, such a field possesses partial field strengths that are not fully gauge invariant, but may be invariant with respect to only a subset of the gauge parameters. The charges associated with these fields therefore involve combinations of partial field strengths, and the corresponding smeared charges inherit this structure. When computing the algebra of smeared charges, we must then evaluate gauge variations with respect to all gauge parameters of the field. As a result, non-vanishing terms may arise from the variation of a partial field strength with respect to a gauge parameter under which it is not invariant. This suggests a direct correlation between the Young symmetry of the mixed-symmetry field, the structure of its partial field strengths, and the emergence of non-abelian features in the algebra of smeared charges.
\section{Conclusions and outlook}

In this work we have studied in a unified manner the duality,
asymptotic charges and generalized symmetries of $p$-form gauge theories in
$D$-dimensional Minkowski spacetime. 

In Section \ref{par2}, mainly containing review material, we begin with a detailed analysis of the Bondi expansion for $p$-form gauge fields and we computed the asymptotic charges under both radiative and Coulombic fall-offs for fields. This produced a universal expression for the electric-like charges \eqref{chargeprad} and \eqref{chargepcoul} in terms of the $H_{rui_1...i_{p-1}}$ components of the field strength, valid for all admissible pairs
$(p,D)$. We then derived, using the Hodge duality equations,
magnetic-like charges for the dual $q$-form theory where $q=D-p-2$. We showed that the two
sectors naturally assemble into a complexified electromagnetic-like charge \eqref{electromagncharge} and that
complex charges transform under duality according to a Möbius transformation
parametrised by an element of $\mathrm{PGL}(2,\mathbb{C})$ given by \eqref{apgl}. This behaviour generalises the
electric-magnetic duality rotation known from $D=4$ Maxwell theory, and suggests a
conceptual bridge between infrared dualities and the conformal structure
characteristic of celestial descriptions of scattering.\\
Motivated by the research of this bridge and going beyond the review material, we propose a novel geometric framework in which celestial
CFTs may be understood as $\mathrm{PGL}(2,\mathbb{C})$-principal equivariant bundles over the
celestial sphere. In this picture, celestial operators arise as sections of associated
bundles, while descendants and OPE structures could be encoded through finite or
infinite jet bundles. Although preliminary, this perspective offers a unified
geometrization of Lorentz symmetry, duality, and conformal weight assignments,
and suggests new tools for organising CCFT spectra. Moreover, the extension of this
bundle-theoretic approach to incorporate spin structures provides a natural setting
for the inclusion of spinorial operators and representations.

A major result, in Section \ref{secTHM} is the proof of an existence and uniqueness theorem for the
duality map relating the asymptotic electric-like charges of the two formulations, Theorem \ref{THM3.1}.
The natural setting for this analysis is the de~Rham complex on Minkowski
spacetime or, more generally, a spacetime satisfying \eqref{cohom}. Within this framework we constructed an explicit linear isomorphism intertwining the charges of the two dual
theories and we showed that this duality map is intrinsically topological in nature.
This observation has several implications: for example, it suggests that in
backgrounds admitting non-trivial cohomology the duality map need not be unique,
and may even fail to exist. Such a breakdown may occur in situations involving
topology changes, a possibility that could be relevant in quantum gravity, where topology is expected to fluctuate.\\
In the same section we discuss the duality problem for diverging/vanishing charges. Note that these divergences could be renormalized by symplectic procedures. On the one hand, the interplay of electric-like and magnetic-like charges at Coulombic order offers the natural picture for the appearance of scalar charges that act at order $\mathcal{O}(r^{-2})$, useful in the description of the pion memory effect. On the other hand, we introduce, in Conjecture \ref{Conj3.1} and motivated by relation \eqref{mapcol}, the idea that if a gauge theory has a divergent asymptotic charge its dual description still has trivial gauge transformations. In view of renormalization issue, since renormalizing a divergent charge amounts to removing divergences that are not physical, this would mean removing
the unphysical (i.e. redundant) gauge transformations from the dual description. 

In Section \ref{sec4} of the work, we established a concrete link between asymptotic charges and the higher-form symmetry charges characteristic of
$p$-form gauge theories. By integrating the field strengths of the two dual
descriptions over suitable codimension submanifolds, we reproduced the standard
higher-form charges, i.e. \eqref{hf}, while smearing the associated local currents with residual
gauge parameters of the dual theory produced a smeared charge, i.e. \eqref{sme}, proportional to the asymptotic one. In particular, our construction addresses part
of the open question concerning the relation between asymptotic and generalized
symmetries, showing that at least for abelian $p$-form theories the two frameworks
are more tightly connected than previously appreciated.

\medskip

\noindent\textbf{Future directions.}
Several paths for further research emerge from the present analysis. Here we present some of them.

First, the topological nature of the duality map invites a study of $p$-form gauge
theories on backgrounds with non-trivial cohomology, where the duality map may
become multi-valued or obstructed. Understanding how this affects asymptotic
charges could provide insight into
infrared behavior in compactified theories and in quantum gravity. One particular interesting background could be that of gravisolitons which have found interest in the swampland conjectures and which have been found to possess asymptotic symmetries different from the background metric.

Second, the connection established here between asymptotic charges and
generalized symmetries suggests that similar mechanisms may exist for
non-abelian higher-form theories or for theories with mixed-symmetry tensor
fields. The analysis of residual gauge transformations and asymptotic
configurations for such more exotic fields remains largely unexplored and may reveal new 
infrared structures. Moreover, a better understanding in cases of gravitational dressing setups has to be taken in consideration. 

Third, the Möbius-equivariant geometrization of celestial CFT developed in this
paper raises several technical questions. It would be valuable to classify the
principal bundle structures compatible with the CCFT operator spectrum, to
investigate how OPE coefficients behave under duality actions on the principal
bundle, and to determine whether the electromagnetic-like duality induces a
non-trivial automorphism on the sheaf of celestial operators, if any. The role of infinite
jet bundles in encoding full Verma modules also deserves further study, especially
in the presence of interactions or loop corrections. As for the classification, in our construction we assumed a topologically trivial bundle. It would be interesting to classify the possible holomorphic structures of such a bundle as well as of its topologically non-trivial case, if suitable for spin lifting.

Finally, extending our results to AdS backgrounds or to asymptotically flat spacetimes
with additional structure (e.g.\ massive modes, non-trivial fluxes, or defects)
would test the robustness of the duality map and may reveal new holographic
signatures. Such generalizations could clarify whether the interplay between
$p$-form dualities, asymptotic symmetries, and generalized charges constitutes a
universal phenomenon or is specific to the simplest flat-space setting.
\bibliographystyle{JHEP} 
\bibliography{References}

\end{document}